\documentclass{amsart}

\usepackage{amssymb,amsthm,hyperref,enumitem,graphicx,array}
\usepackage[round,comma]{natbib}

\usepackage[left=1.5in,right=1.5in,top=1.5in,bottom=1.5in]{geometry}

\graphicspath{{./Figures/}}

\newcommand{\model}{\mathcal{M}}

\newcommand{\SBIC}{\mathrm{sBIC}}

\newtheorem{theorem}{Theorem}

\newtheorem{proposition}{Proposition}

\theoremstyle{definition}
\newtheorem{definition}{Definition}

\theoremstyle{remark}
\newtheorem{remark}{Remark}
\newtheorem{example}{Example}

\numberwithin{equation}{section}
\numberwithin{figure}{section}
\numberwithin{theorem}{section}
\numberwithin{lemma}{section}
\numberwithin{proposition}{section}
\numberwithin{corollary}{section}
\numberwithin{question}{section}
\numberwithin{definition}{section}
\numberwithin{xca}{section}
\numberwithin{remark}{section}
\numberwithin{example}{section}

\title{A Bayesian information criterion for singular models}

\author{Mathias Drton} 
\address{Department of Statistics, University
  of Washington, Seattle, WA 98195-4322} 
\email{md5@uw.edu}

\author{Martyn Plummer}
\address{International Agency for Research on Cancer, 150 cours Albert
  Thomas, 69372 Lyon cedex 08, France}
\email{plummerM@iarc.fr}

\begin{document}

\begin{abstract}
  We consider approximate Bayesian model choice for model selection
  problems that involve models whose Fisher-information matrices may
  fail to be invertible along other competing submodels.  Such
  singular models do not obey the regularity conditions underlying the
  derivation of Schwarz's Bayesian information criterion (BIC) and the
  penalty structure in BIC generally does not reflect the frequentist
  large-sample behavior of their marginal likelihood.  While
  large-sample theory for the marginal likelihood of singular models
  has been developed recently, the resulting approximations depend on
  the true parameter value and lead to a paradox of circular
  reasoning.  Guided by examples such as determining the number of
  components of mixture models, the number of factors in latent factor
  models or the rank in reduced-rank regression, we propose a
  resolution to this paradox and give a practical extension of BIC for
  singular model selection problems.
\end{abstract}


\maketitle

\section{Introduction}
\label{sec:intro}

Information criteria are valuable tools for model selection
\citep{Burnham:2002,Claeskens:2008,Konishi:2008}.  At a high level,
they fall into two categories
\citep{yang:2005,vanerven:2012,wit:2012}.  On one hand, there are
criteria that target good predictive behavior of the selected model.
For instance, cross-validation based scores assess the quality of
out-of-sample predictions by splitting available data into test
and training cases, and Akaike's information criterion (AIC) provides
an estimate of an out-of-sample prediction (or generalization) error
that is justified via asymptotic distribution theory for large samples
\citep{Akaike:1974}.  Following a different philosophy that will be
the focus of this paper, the Bayesian information criterion (BIC) of
\cite{Schwarz:1978} draws motivation from Bayesian inference.
Schwarz's criterion aims to capture key features of
posterior model uncertainty via a penalty that is motivated by the
large-sample properties of the marginal likelihood (also commonly
referred to as integrated likelihood or evidence).  In a nutshell,
under suitable regularity conditions, a quadratic approximation to the
log-likelihood function can be used to relate the marginal likelihood
to a Gaussian integral in which the sample size acts as an inverse
variance.  This dependence of the Gaussian integral on the sample size
leads to the familiar BIC penalty term that, on the log-scale,
consists of the product of model dimension and the logarithm of the
sample size.

The BIC penalizes model complexity more heavily than predictive
criteria such as AIC.  From the frequentist perspective, it has been
shown that the BIC's penalty depends on the sample size in a way that
makes the criterion consistent for a wide range of problems.  In other
words, when optimizing BIC the probability of selecting a fixed
most parsimonious true model tends to one as the sample size tends to
infinity \cite[e.g.,][]{Nishii:1984,Haughton:1988,Haughton:1989}.
However, a wide range of penalties would yield consistency of a model
selection score, and it is instead the aim of capturing the asymptotic
scaling of the marginal likelihood that leads to the familiar
dependence on dimension and log-sample size.  Indeed, from a
Bayesian point of view, the BIC supplies rather crude but
computationally inexpensive proxies to otherwise difficult to
calculate posterior model probabilities, which form the basis for
Bayesian model choice and averaging; see \cite{Kass:1995},
\cite{raftery:1995}, \cite{DiCiccio:1997}, \cite{hoeting:1999} or
\citet[][Chap.~7.7]{hastie:2009}.

In this paper, we are concerned with Bayesian information criteria in
the context of singular model selection problems, that is, problems
that involve models with Fisher-information matrices that may fail to
be invertible.  For example, due to the break-down of parameter
identifiability, the Fisher-information matrix of a mixture model with
three components is singular at a distribution that can
be obtained by mixing only two components.  This clearly presents a
fundamental challenge for selection of the number of components.  In
particular, when the Fisher-information matrix is singular, the
log-likelihood function does not admit a large-sample approximation by
a quadratic form.  \cite{Rotnitzky:2000} illustrate some of the
resulting difficulties in asymptotic distribution theory under an
assumption of identifiability.  Non-identifiability of parameters, as
present in the examples we will consider, leads to considerably more
complicated scenarios as discussed, for instance, by \cite{Liu:2003} and
\cite{azais:2006,azais:2009}.  
The key obstruction to justifying BIC is that in singular models there
need no longer be a connection between the Bayesian marginal
likelihood and a Gaussian integral.  In particular, a parameter count
or model dimension may fail to capture the asymptotic scaling of the
marginal likelihood \citep{Watanabe:Book}.  We illustrate this fact in
the following example.

\begin{example}
  \label{ex:intro}
  Suppose $\mathbf{Y}_n=(Y_{n1},\dots,Y_{nn})$ is a sample of independent and
  identically distributed observations whose unknown distribution is
  modeled as a mixture of two normal distributions.  Specifically,
  the data-generating distribution is assumed to be of the form
  \[
  \pi(\alpha,\mu_1,\mu_2) \;:=\; \alpha \mathcal{N}(\mu_1,1) +
  (1-\alpha) \mathcal{N}(\mu_2,1), 
  \]
  where $\alpha\in[0,1]$ is an unknown mixture weight,
  $\mu_1,\mu_2\in\mathbb{R}$ are two unknown means, and the variances
  are known and equal to one.  To exemplify later notation, we write
  out the likelihood function of the considered mixture model
  $\mathcal{M}$, which  maps the parameter vector $(\alpha,\mu_1,\mu_2)$ to
  \[
  P(\mathbf{Y}_n\,|\,\pi(\alpha,\mu_1,\mu_2),\mathcal{M}) \;=\; \prod_{i=1}^n \left[
    \alpha \varphi\big(Y_{ni}-\mu_1\big) + (1-\alpha)
    \varphi\big(Y_{ni}-\mu_2\big) 
  \right].
  \]
  Here, $\varphi$ denotes the standard normal density.  As a prior for
  Bayesian inference, consider a uniform distribution for $\alpha$,
  and take $\mu_1$ and $\mu_2$ to be independent $\mathcal{N}(0,16)$.
  Then the marginal likelihood of model $\mathcal{M}$ is
  \[
  L(\mathcal{M}) = \int_{[0,1]\times \mathbb{R}^2}
  P(\mathbf{Y}_n\,|\,\pi(\alpha,\mu_1,\mu_2),\mathcal{M})\,
  \varphi(\mu_1/4)\,\varphi(\mu_2/4) \; d(\alpha,\mu_1,\mu_2).
  \] 

  We now simulate values of the random variable $L(\mathcal{M})$.  For
  each choice of a sample size $n\in\{50,60,\dots,100\}$, we generate
  200 independent realizations of $L(\mathcal{M})$, drawing the sample
  $\mathbf{Y}_n$ from the normal mixture $\pi_0$ given by
  $\alpha=0.4$, $\mu_1=-2$ and $\mu_2=2$.  Following \cite{neal:1999},
  we compute each value of $L(\mathcal{M})$ by standard Monte Carlo
  with $10^7$ draws from the prior.  To allow for comparisons across
  different samples, we consider the marginal likelihood ratio
  $L_0(\mathcal{M})$ that is obtained by dividing $L(\mathcal{M})$ by
  $P(\mathbf{Y}_n\,|\,\pi_0)$, the likelihood of the sample under the
  true distribution.  The results are summarized in
  Figure~\ref{fig:intro}(a), which plots average values of
  $2\log L_0(\mathcal{M})$
  together with a least squares line relating $2\log L_0(\mathcal{M})$
  to $\log(n)$.  We also show 90\%-simultaneous confidence bands and a
  line with slope determined by large-sample theory.  We emphasize
  that the figure's horizontal axis has the sample size on the
  log-scale.  The slope of the least squares line comes out to be
  $-2.98$ and is close to the slope of $-3$ that is predicted by the
  parameter count from Schwarz's BIC.

  A different picture emerges, however, when we repeat the simulations
  changing the data-generating distribution $\pi_0$ to the standard
  normal distribution $\mathcal{N}(0,1)$; see
  Figure~\ref{fig:intro}(b).  In this case, the slope of the least
  squares line is no longer close to the negated parameter count.
  Instead, it is about $-1.62$.  In Section~\ref{sec:background}, we
  discuss asymptotic theory that addresses the issue that the
  Fisher-information matrix of $\mathcal{M}$ is singular at the
  standard normal distribution.  For $\mathcal{N}(0,1)$ data in this
  example, the theory predicts a slope of $-1.5$ \citep{aoyagi:2010}.
  This large-sample line is contained in the simultaneous confidence
  bands we give in Figure~\ref{fig:intro}(b).

  As we will review in Section~\ref{sec:background}, refined
  mathematical knowledge about the asymptotic scaling of the marginal
  likelihood of singular models has been obtained in recent years.  It
  is desirable to leverage this knowledge when defining an information
  criterion that is inspired by Bayesian methods.  However, it is not
  immediately clear how to cope with the fact that even the most basic
  features of the asymptotics for the marginal likelihood depend on
  the unknown data-generating distribution.  The generalization of BIC
  we introduce in this paper resolves this issue by averaging
  different approximations in a data-dependent way. 
\end{example}

\begin{figure}[t]
  \centering
  (a) \includegraphics[width=0.45\textwidth]{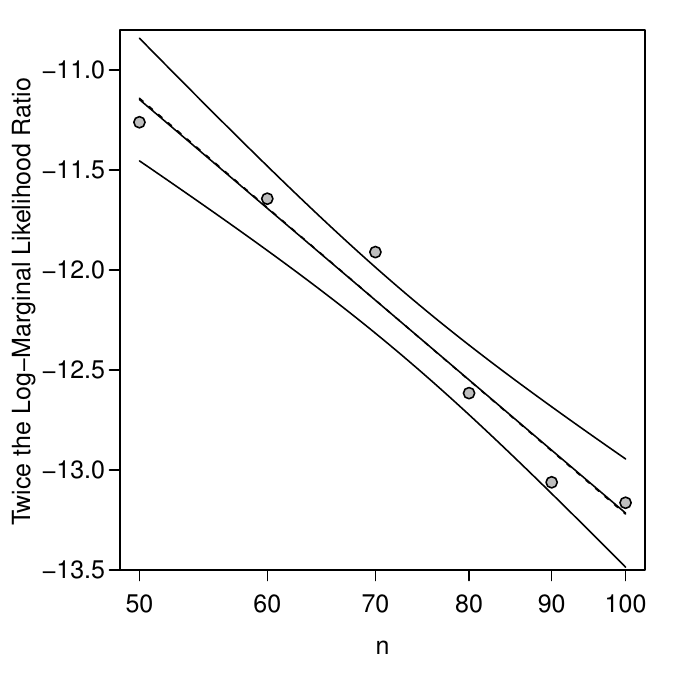}
  (b) \includegraphics[width=0.45\textwidth]{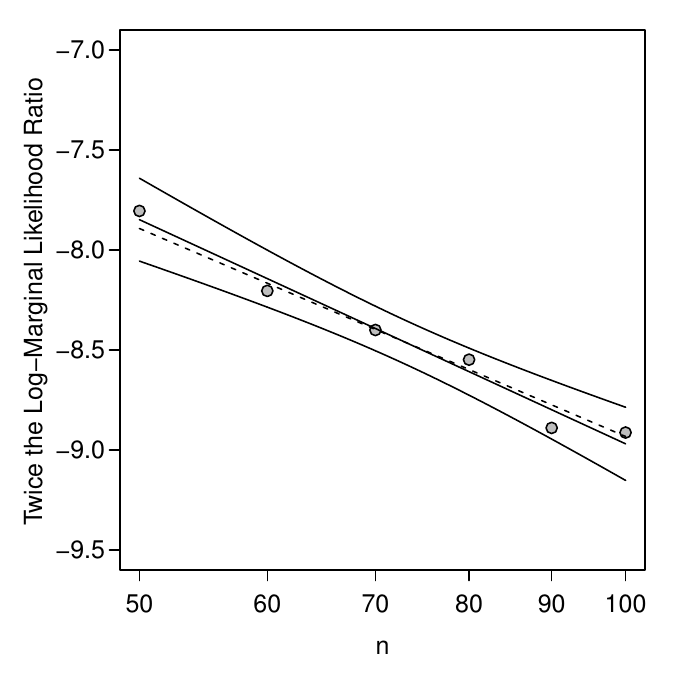}
  \caption{Averages of twice the log-marginal likelihood ratio
    for a Gaussian mixture model, least squares line, simultaneous
    confidence bands, and a line with theory-based slope (dashed): (a)
    data from a two-component mixture; (b) standard normal data. }
  \label{fig:intro}
\end{figure}

As conveyed by the above example, the selection of the number of
mixture components constitutes a singular model selection problem.
Other important examples of this type include determining the rank in
reduced-rank regression, the number of factors in factor analysis or
the number of states in latent class or hidden Markov models.  More
generally, all the classical hidden/latent variable models are
singular, which expresses itself also in complicated geometry of the
parameter space/set of distributions
\citep{geiger:2001,drton:2007,zwiernik:2012,zwiernik:2014,gassiat:2014}.

Despite the possible disconnect between penalization based on model
dimension alone and the large-sample behavior of Bayesian methods, the
standard BIC is a state-of-the-art method for many singular model
selection problems; see e.g.~\citet[Section 6.9]{McLachlan:2000},
\cite{steele:2010} and \cite{Baudry:2015} for mixture models and
\cite{Lopes:2004} for factor analysis.  From the frequentist
perspective, BIC is known to be consistent in many singular settings
\citep[Chap.~5.1]{keribin:2000,Drton:Book}.  However, as mentioned
earlier, consistency can be achieved with many penalization schemes,
which would not need to 
depend logarithmically on the sample size.

In this paper, we propose a generalization of BIC that utilizes
refined mathematical information about the marginal likelihood of the
considered statistical models, information that goes beyond mere model
dimension.  Schwarz's BIC is Bayesian in the sense that it differs
from the log-marginal likelihood only by terms that are bounded.  The
new criterion, which we abbreviate to $\SBIC$, maintains this
connection to Bayesian model choice also in singular settings.  Our
$\SBIC$ criterion preserves consistency properties of BIC and is an
honest generalization of the standard criterion in the sense that
$\SBIC$ coincides with Schwarz's BIC when the model is regular.
$\SBIC$ is designed to capture the key features of posterior model
uncertainty, but our numerical work shows that it can also lead to
improved frequentist model selection properties.

The new criterion is presented in Section~\ref{sec:bic:sing}, which is
preceded by a review of the theory that $\SBIC$ is built on
(Section~\ref{sec:background}).  This theory was developed over the
last decade by \cite{watanabe:2001,Watanabe:Book}.  The large-sample
properties of $\SBIC$ are shown in Section~\ref{sec:large-sample}.  We
first show consistency and then clarify the connection to the
log-marginal likelihood.  Numerical examples demonstrating the use of
$\SBIC$ are given in Sections~\ref{sec:numerical}
and~\ref{sec:numerical-mix}.  The former section focuses on problems
from multivariate analysis, namely, reduced-rank regression and factor
analysis.  The latter section treats mixture models, where it becomes
particularly apparent that the choice of the prior distribution in
each (singular) model has an impact on the form of $\SBIC$---as a
reader familiar with the work of \cite{rousseach:2011} may suspect.
We conclude the paper with a discussion of the strengths and the
limitations of the proposed methodology in
Section~\ref{sec:conclusion}.

\section{Background}
\label{sec:background}

Let $\mathbf{Y}_n=(Y_{n1},\dots,Y_{nn})$ denote a sample of $n$
independent and identically distributed observations, and let
$\{\model_i : i\in I\}$ be a finite set of candidate models for the
distribution of these observations.  For a Bayesian treatment, suppose
that we have positive prior probabilities $P(\model_i)$ for the
different models and that, in each model $\model_i$, a prior
distribution $P(\pi_i \,|\,\model_i)$ is specified for the probability
distributions $\pi_i\in\model_i$.  Write
$P(\mathbf{Y}_n\,|\, \pi_i,\model_i)$ for the likelihood of
$\mathbf{Y}_n$ under data-generating distribution $\pi_i$ from model
$\model_i$, and  let
\begin{equation}
  \label{eq:marg-lik}
  L(\model_i)\; := \; P(\mathbf{Y}_n\,|\, \model_i) = \int_{\model_i} P(\mathbf{Y}_n\,|\,
  \pi_i, \model_i)\; dP(\pi_i\,|\, \model_i)
\end{equation}
be the marginal likelihood of model $\model_i$.  Bayesian model choice
is then based on the posterior model probabilities
\[
P(\model_i\,|\, \mathbf{Y}_n) \;\propto\; P(\model_i) L(\model_i), \quad i\in I.
\]

The probabilities $P(\model_i\,|\, \mathbf{Y}_n)$ can be approximated by
various Monte Carlo procedures---see \cite{friel:2012} for a recent
review---but practitioners also often turn to computationally
inexpensive proxies suggested by large-sample theory.  These proxies
are based on the asymptotic properties of the sequence of random
variables $L(\model_i)$ obtained when $\mathbf{Y}_n$ is drawn from a
data-generating distribution $\pi_0\in \model_i$, and we let the
sample size $n$ grow.

In practice, a prior distribution $P(\pi_i\,|\, \model_i)$ is
typically specified by placing a distribution on the vector of parameters
appearing in a parametrization of $\model_i$; recall
Example~\ref{ex:intro}.  So assume that
\begin{equation}
  \label{eq:parametrized-model}
  \model_i =\left\{\, \pi_i(\boldsymbol{\omega}_i) : \boldsymbol{\omega}_i\in\Omega_i\,\right\}
\end{equation}
with $d_i$-dimensional parameter space
$\Omega_i\subseteq\mathbb{R}^{d_i}$, and that $P(\pi_i\,|\, \model_i)$
is the transformation of a distribution
$P(\boldsymbol{\omega}_i\,|\,\model_i)$ on $\Omega_i$ under the map
$\boldsymbol{\omega}_i\mapsto \pi_i(\boldsymbol{\omega}_i)$.  The
marginal likelihood is then the $d_i$-dimensional integral
\begin{equation}
  \label{eq:marg-lik-para}
  L(\model_i)\; = \; \int_{\Omega_i} P(\mathbf{Y}_n\,|\, \pi_i(\boldsymbol{\omega}_i), \model_i)\;
  dP(\boldsymbol{\omega}_i\,|\, \model_i).
\end{equation}
Now the observation of Schwarz and other subsequent work is that, under
suitable technical conditions on the model $\model_i$, the
parametrization $\boldsymbol{\omega}_i\mapsto \pi_i(\boldsymbol{\omega}_i)$ and the prior
distribution $P(\boldsymbol{\omega}_i\,|\,\model_i)$, it holds for all $\pi_0\in
\model_i$ that
\begin{equation}
  \label{eq:bic}
  \log L(\model_i) \;=\; \log P(\mathbf{Y}_n\,|\, \hat\pi_i,\model_i) - \frac{d_i}{2}
  \log(n) + O_p(1).
\end{equation}
Here, $P(\mathbf{Y}_n\,|\,\hat\pi_i,\model_i)$ is the maximum of the likelihood
function, 
and $O_p(1)$ stands for a sequence of remainder terms that is bounded
in probability. 
The first two terms on the right-hand side of (\ref{eq:bic}) can be
evaluated in statistical practice and 
may be used as a model score or a proxy for the logarithm of the
marginal likelihood.  The resulting \emph{Bayesian} or \emph{Schwarz's
  information criterion} for model $\model_i$ is 
\begin{equation}
  \label{def:bic}
  \mathrm{BIC}(\model_i) \;=\; \log P(\mathbf{Y}_n\,|\, \hat\pi_i,\model_i) - \frac{d_i}{2}
  \log(n).
\end{equation}

Briefly put, the large-sample behavior from (\ref{eq:bic}) relies on
the following properties of regular problems.  First, with high
probability, the integrand in (\ref{eq:marg-lik-para}) is negligibly
small outside a neighborhood of the maximum likelihood estimator
of $\boldsymbol{\omega}_i$.  Second, in such a neighborhood, the log-likelihood
function $\log P(\mathbf{Y}_n\,|\,\pi_i(\boldsymbol{\omega}_i),\model_i)$ can be
approximated by a negative definite quadratic form, while a smooth
prior $P(\boldsymbol{\omega}_i\,|\,\model_i)$ is approximately constant.  The
integral in~(\ref{eq:marg-lik-para}) may thus be approximated by the
product of $P(\mathbf{Y}_n\,|\, \hat\pi_i,\model_i)$ and a Gaussian
integral, in which the inverse covariance matrix equals $n$ times the
Fisher-information.  This $d_i$-dimensional Gaussian integral depends
on $n$ via the multiplicative factor $n^{-d_i/2}$, and taking
logarithms one arrives at~(\ref{eq:bic}).  We remark that this approach also
allows for estimation of the remainder term in~(\ref{eq:bic}), giving
a Laplace approximation with error $O_p(n^{-1/2})$ as discussed, for instance, in
\citet{Tierney:1986,Haughton:1988,Kass:1995}, or
\cite{Wasserman:2000}.

A large-sample quadratic approximation to the log-likelihood function
is not possible, however, when the Fisher-information matrix is
singular.  Consequently, the classical theory alluded to above does
not apply to singular models.  Indeed, (\ref{eq:bic}) is generally
false in singular models.  Nevertheless, asymptotic theory for the
marginal likelihood of singular models has been developed over the
last decade, culminating in the monograph of \cite{Watanabe:Book}.
Indeed, Theorem 6.7 in \cite{Watanabe:Book} shows that a wide variety
of singular models have the property that, for $\mathbf{Y}_n$ drawn from
$\pi_{0}\in \model_i$,
\begin{equation}
  \label{eq:wata-pi0}
  \log L(\model_i) \;=\; 
  \log P(\mathbf{Y}_n\,|\, \pi_0,\model_i) - \lambda_i(\pi_{0}) \log(n) + \big[
  m_i(\pi_0)-1 \big] \log\log(n) + O_p(1);
\end{equation}
see also the introduction to the topic in
\citet[Chap.~5.1]{Drton:Book}.  In this paper, we follow the
terminology of \cite{Watanabe:Book} and refer to the quantity
$\lambda_i(\pi_0)$ as the \emph{learning coefficient}.  However, other
terminology such as \emph{real log-canonical threshold} or
\emph{stochastic complexity} is in use.  The number $m_i(\pi_0)$ is
the \emph{multiplicity} of the learning coefficient/real log-canonical
threshold.  In contrast to the regular case, it is generally very
difficult to estimate the $O_p(1)$ remainder term
in~(\ref{eq:wata-pi0}).  We are not aware of any successful work on
higher-order approximations in statistically relevant singular
settings.

\begin{remark}
  The theorem giving (\ref{eq:wata-pi0}) is developed under the
  `fundamental conditions (I) and (II)' from Definitions 6.1 and 6.3
  in \cite{Watanabe:Book}.  While the precise nature of these
  conditions is not important for the developments in this paper, we
  would like to summarize them briefly.  Under 
  conditions (I) and (II), the distributions in $\model_i$ share a
  common support and have densities with respect to a dominating
  measure.  The parameter space $\Omega_i$
  in~(\ref{eq:parametrized-model}) is compact and defined by real
  analytic constraints.  (An assumption of compactness is only needed
  when the set of parameter vectors representing the true distribution
  $ \{ \boldsymbol{\omega}_i \in \Omega_i :
  \pi(\boldsymbol{\omega}_i)=\pi_0 \} $ is not already compact.)
  Watanabe's conditions further require that the log-likelihood ratios
  of $\pi_0$ with respect to the distributions
  $\pi(\boldsymbol{\omega}_i)$ can be bounded by a function that is
  square-integrable under $\pi_0$.  Moreover, the log-likelihood
  ratios satisfy a requirement of analyticity that allows for power
  series expansions in $\boldsymbol{\omega}_i$.  Finally, the prior
  distribution $P(\boldsymbol{\omega}_i\,|\,\model_i)$ has a density
  that is the product of a smooth positive function and a nonnegative
  analytic function.
\end{remark}

Watanabe's result applies to models such as reduced-rank regression,
factor analysis, Binomial mixtures, and latent class analysis, which
we will consider in the numerical experiments of
Sections~\ref{sec:numerical} and~\ref{sec:numerical-mix}.  Via
suitable analytic bounds on the log-likelihood ratios, the result can
also be extended to other `non-analytic models', such as mixtures of
normal distributions with known common variance as we considered in
Example~\ref{ex:intro} \cite[Section 7.8]{Watanabe:Book}.  While the
case of Gaussian mixtures with unknown variance has not yet been
treated explicitly in the literature, we show experiments with such
models in Section~\ref{sec:numerical-mix}.

\begin{example}
  \label{ex:intro2}
  Let $\model_2$ be the Gaussian mixture model with $i=2$ components
  that we considered in Example~\ref{ex:intro}.  If $\pi_0$ is a
  normal distribution $\mathcal{N}(\mu,1)$ then
  $\lambda_2(\pi_0)=3/4$.  If $\pi_0$ is an honest mixture of two
  normal distributions with variance 1 then $\lambda_2(\pi_0)=3/2$.
  In either case $m_2(\pi_0)=1$.  The values can be found in Example
  3.1 of \cite{aoyagi:2010}.\footnote{The formula for the multiplicity
    in Theorem 3.4 in \cite{aoyagi:2010} applies only if $r<H$, in the
    notation used there.  If $r=H$, the multiplicity is one, as
    confirmed in private communication with the author.}
\end{example}

Reduced-rank regression, factor analysis and latent class analysis are
all singular submodels of an exponential family, which is either the
normal or the multinomial family.  It follows that the sequence of
likelihood ratios
$P(\mathbf{Y}_n\,|\, \hat\pi_i,\model_i)/P(\mathbf{Y}_n\,|\,
\pi_0,\model_i)$
converges in distribution and, in particular, is bounded in
probability \citep{drton:2009}.  In this case, we can plug the maximum
likelihood estimator into the first term of~(\ref{eq:wata-pi0}) and
obtain that
\begin{multline}
  \label{eq:wata}
  \log L(\model_i) \;=\; 
  \log P(\mathbf{Y}_n\,|\, \hat\pi_i,\model_i) - \lambda_i(\pi_{0}) \log(n) +
  \big[ m_i(\pi_0)-1 \big] \log\log(n) + O_p(1).
\end{multline}
For more complicated models, such as Gaussian mixture models,
likelihood ratios can often be shown to converge in distribution under
compactness assumptions on the parameter space; see for instance
\cite{azais:2006,azais:2009} who also review much of the relevant
literature.  Without compactness, the log-likelihood ratios in mixture
models need not be bounded in probability; e.g., they would be of
order $O_p(\log\log(n))$ for Gaussian mixtures
\citep{hartigan:1985,bickel:1993}.

Having estimated the log-likelihood by estimating the unknown
data-generating distribution $\pi_0$, it seems tempting to similarly
estimate the learning coefficient $\lambda_i(\pi_0)$ and its
multiplicity $m_i(\pi_0)$.  However, in contrast to the likelihood
function, the learning coefficient and multiplicity are not continuous
functions of $\pi_0$.  Hence, substituting an estimate for $\pi_0$ is
of little interest as the resulting expression fails to capture the
behavior of the marginal likelihood at (or near) model singularities.
Instead, we will make the fact from~(\ref{eq:wata}) the point of
departure in the definition of our singular BIC, which is the topic of
Section~\ref{sec:bic:sing}.

As in the original work of \cite{Schwarz:1978}, our general treatment
will focus on prior distributions with smooth densities that are
bounded and positive.  On a compact set, such a density will be
bounded away from zero. In the analytic settings considered in
\cite{Watanabe:Book}, it then holds that $\lambda_i(\pi_0)$ is a
rational number in $[0,d_i/2]$ and $m_i(\pi_0)$ is an integer in
$\{1,\dots, d_i\}$.  However, as mentioned above, priors with
densities that are zero in parts of the parameter space can be
accommodated in the framework as long as the prior
density vanishes in an `analytic fashion'.
In this case, the learning coefficient may depend on the prior
$P(\boldsymbol{\omega}_i\,|\,\model_i)$ in important ways.  In
particular, if the prior has a density that is zero at model
singularities then $\lambda_i(\pi_0)$ could exceed $d_i/2$; compare
the discussion of Jeffrey's prior in Theorem 7.4 in
\cite{Watanabe:Book}.  We will revisit the role of the prior
distribution in experiments with mixture models in
Sections~\ref{sec:gauss-mix} and~\ref{subsec:lca}.

\begin{example}
  \label{ex:red-reg}
  Reduced-rank regression is multivariate linear regression subject to
  a rank constraint on the matrix of regression coefficients
  \citep{reinsel:1998}.  Suppose we observe $n$ independent copies of
  a partitioned zero-mean Gaussian random vector $Y=(Y_R,Y_C)$, where
  $Y_R\in\mathbb{R}^N$ and $Y_C\in\mathbb{R}^M$.  Keeping only with
  the most essential structure, assume that the covariance matrix of
  $Y_C$ and the conditional covariance matrix of $Y_R$ given $Y_C$ are
  both the identity matrix.  The reduced-rank regression model
  $\model_i$ associated to an integer $i\ge 0$ then postulates that
  the $N\times M$ matrix $\pi$ in the conditional expectation
  $\mathbb{E}[Y_R\,|\,Y_C]=\pi Y_C$ has rank at most $i$.
  
  In a Bayesian treatment, consider the parametrization
  $\pi=\boldsymbol{\omega}_{i2}\boldsymbol{\omega}_{i1}$, with smooth
  and positive prior densities for
  $\boldsymbol{\omega}_{i2}\in\mathbb{R}^{N\times i}$ and
  $\boldsymbol{\omega}_{i1}\in\mathbb{R}^{i\times M}$.  Note that
  while the matrix $\pi$ is in one-to-one correspondence with the
  joint distribution of $Y$, this is not true for the pair of matrices
  $\boldsymbol{\omega}_i=(\boldsymbol{\omega}_{i1},\boldsymbol{\omega}_{i2})$
  used to parametrize the model.  For this setup,
  \cite{aoyagi-watanabe:2005} derived the learning coefficients
  $\lambda_i(\pi_0)$ and their multiplicities $m_i(\pi_0)$, where the
  true data-generating distribution is given by an $N\times M$ matrix
  $\pi_0$ of rank $j\le i$.
  In particular, $\lambda_i(\pi_0)$ and $m_i(\pi_0)$ depend on $\pi_0$
  only through the true rank $j$.  

  \begin{table}
    \centering
    \caption{Learning coefficients for reduced-rank regression ($5$
      responses, $3$ covariates): model postulates rank 
      $i$,  true rank is $j$.} 
    \label{tab:learn-coeffs-redreg}
    \begin{tabular}[t]{p{12cm}}
      \[
      \renewcommand\arraystretch{1.2}
      \begin{array}{rcccc}
        \hline\hline
        & j=0 & j=1 & j=2 & j=3\\
        \hline
        i=0 & 0\\
        i=1 & \frac{3}{2} & \frac{7}{2}           & \\
        i=2 & 3 & \frac{9}{2} & 6  \\
        i=3 & \frac{9}{2}  & \frac{11}{2}   &   \frac{13}{2}      
                          & \frac{15}{2}
                            \\
        \hline
        \hline
      \end{array}
      \]
    \end{tabular}
  \end{table}

  For a concrete instance, take $N=5$ and $M=3$.  Then the values of
  $\lambda_{ij}:= \lambda_i(\pi_0)$ are listed in
  Table~\ref{tab:learn-coeffs-redreg}, and the multiplicity
  $m_i(\pi_0)=1$ unless $i=3$ and $j=0$ in which case $m_i(\pi_0)=2$.
  Note that the table entries for $j=i$ are equal to
  $\dim(\model_i)/2$, where $\dim(\model_i)=i(N+M-i)$ is the dimension
  of $\model_i$, which can be identified with the set of $N\times M$
  matrices of rank at most $i$.  The dimension is also the maximal
  rank of the Jacobian of the map
  $(\boldsymbol{\omega}_{i1},\boldsymbol{\omega}_{i2})\mapsto
  \boldsymbol{\omega}_{i2}\boldsymbol{\omega}_{i1}$.
  The singularity issues addressed in Watanabe's theory arise at
  points where the Jacobian of the parametrization fails to have
  maximal rank.  These have
  $\text{rank}(\boldsymbol{\omega}_{i2}\boldsymbol{\omega}_{i1})<i$
  and thus define a distribution that also belongs to a submodel
  $\model_j\subset \model_i$ given by a lower rank $j<i$.  This
  presents a challenge for model selection, which here amounts to
  selection of an appropriate rank.
  
  While the regularity conditions in its derivation are not met, it is
  common practice to apply the standard BIC for selection of the rank.
  In doing so, one typically takes $d_i=\dim(\model_i)$
  in~\eqref{def:bic}.  Simulation studies on rank selection have shown
  that this criterion has a tendency to favor overly small ranks; for
  a recent example see \cite{cheng:2012}.  The quoted values of
  $\lambda_i(\pi_0)$ give a theoretical explanation for this empirical
  phenomenon, as the use of dimension in BIC
  leads to overpenalization of models that contain the true
  data-generating distribution but are not minimal in that regard.
\end{example}

In other models, determining learning coefficients can be a
challenging problem, but progress has been made.  For some of the
examples that have been treated, we refer the reader to
\cite{aoyagi:2010,aoyagi:2010:2,aoyagi:2009,drton:lin:weihs:zwiernik:2014,rusakov:2005,watanabe-amari:2003,watanabe:watanabe:2007,yamazaki-watanabe:2003,yamazaki:2005,yamazaki-watanabe:2004},
and \cite{zwiernik:2011}.
The use of techniques from computational algebra and combinatorics is
emphasized in \cite{Lin:2011}; see also
\cite{arnold-vol2:1988,vasilev:1979}.

Progress in large-sample theory, however, does not readily translate
into practical statistical methodology because one faces the obstacle
that the learning coefficients depend on the unknown data-generating
distribution $\pi_0$, as indicated in our notation in~(\ref{eq:wata}).
For instance, for the problem of selecting the rank in reduced-rank
regression (Example~\ref{ex:red-reg}), the Bayesian measure of model
complexity that is given by the learning coefficient and its
multiplicity depends on the rank we wish to determine in the first
place; recall also Example~\ref{ex:intro} that is about a mixture
model.  It is for this reason that there is currently no statistical
method that takes advantage of theoretical knowledge about the values
of learning coefficients.  In the remainder of this paper, we propose
a solution for how to overcome the problem of circular reasoning and
give a practical extension of the Bayesian information criterion to
singular models.

\section{New Bayesian information criterion for singular models}
\label{sec:bic:sing}

\subsection{Averaging approximations}

As previously stated, our point of departure is the large-sample
result from (\ref{eq:wata}).  If the learning coefficient
$\lambda_i(\pi_0)$ and its multiplicity $m_i(\pi_0)$ that appear in
this equation were known, then we could directly adopt the ideas in
\cite{Schwarz:1978}, omit the remainder term in~(\ref{eq:wata}), and
define a proxy for the marginal likelihood $L(\model_i)$ as
\begin{equation}
  \label{eq:watanabe-approx-pi0}
  L'_{\pi_0}(\model_i) := P(\mathbf{Y}_n\,|\, \hat\pi_i,\model_i)\cdot 
  \frac{(\log n)^{ m_i(\pi_0)-1 }}{n^{\lambda_i(\pi_0)}}.
\end{equation}
However, in practice, $\lambda_i(\pi_0)$ and $m_i(\pi_0)$ are unknown.
We thus propose to apply standard Bayesian thinking and average the
different possible approximations $L'_{\pi_0}(\model_i)$
from~(\ref{eq:watanabe-approx-pi0}) by assigning a probability
measure $Q_i$ to the distributions in model $\model_i$.  In other
words, we eliminate the unknown distribution $\pi_0$ by
marginalization and compute an approximation to $L(\model_i)$ as
\begin{equation}
  \label{eq:BIC-Q}
  L'_{Q_i}(\model_i) := 
  \int_{\model_i} L'_{\pi_0}(\model_i) \;dQ_i(\pi_0).
\end{equation}
The crux of the matter now becomes choosing an appropriate probability
measure $Q_i$.  

\begin{remark}
  \label{prop:sbic=bic}
  Before discussing particular choices for $Q_i$, we stress that any
  choice for $Q_i$ reduces to Schwarz's criterion in the regular case.
  Here, regularity refers to the setting in which model $\model_i$ with
  parametrization
  $\boldsymbol{\omega}_i\mapsto \pi_i(\boldsymbol{\omega}_i)$ has a
  Fisher-information matrix that is invertible at all
  $\boldsymbol{\omega_i}$ in the parameter space $\Omega_i$.  For a
  model with $d_i$ parameters, it then holds that
  $\lambda_i(\pi_0) ={d_i}/{2}$ and $m_i(\pi_0) = 1$ for all
  data-generating distributions $\pi_0\in\model_i$.  Hence,
  $L'_{\pi_0}(\model_i) = e^{\mathrm{BIC}(\model_i) }$ for all
  $\pi_0\in\model_i$.  The integrand in~(\ref{eq:BIC-Q}) being
  constant we have
  \[
  \log L'_{Q_i}(\model_i) = \mathrm{BIC}(\model_i) 
  \]
  irrespective of the choice of $Q_i$.  
\end{remark}

Returning to the singular case, one possible candidate for $Q_i$ is
$P(\pi_0\,|\, \model_i,\mathbf{Y}_n)$, the posterior distribution in
$\model_i$.  Under this distribution, however, the singular models
encountered in practice have the learning coefficient
$\lambda_i(\pi_0)$ almost surely equal to $\dim(\model_i)/2$ with
multiplicity $m_i(\pi_0)=1$; recall
Example~\ref{ex:red-reg}.\footnote{We assume the set $\model_i$
  corresponds to a subset of Euclidean space with well-defined
  dimension.}  We obtain that
\[
\log L'_{Q_i}(\model_i) =  \log P(\mathbf{Y}_n\,|\, \hat\pi_i,\model_i) -
  \frac{\dim(\model_i)}{2}   \log(n),
\]
which is the usual BIC, albeit with the possibility that
$\dim(\model_i)<d_i$, where $d_i$ is the dimension of the parameter
space $\Omega_i$ when $\model_i$ is presented as
in~(\ref{eq:parametrized-model}).  From a pragmatic point of view,
this choice of $Q_i$ is not attractive as it merely recovers the
adjustment from $d_i$ to $\dim(\model_i)$ that is standard practice
when applying Schwarz's BIC to singular models.  More importantly,
however, averaging with respect to the posterior distribution
$P(\pi_0\,|\, \model_i,\mathbf{Y}_n)$ involves conditioning on the
single model $\model_i$, which clearly ignores the model uncertainty
inherent to model selection problems.

In most practical problems, the finite set of models
$\{\model_i : i\in I\}$ has interesting structure with respect to the
partial order given by inclusion.\footnote{In the examples we consider
  in this paper the order is always a total order.  For an example
  where this is not the case, see
  \cite{drton:lin:weihs:zwiernik:2014}.}  For notational convenience,
we define the poset structure on the index set $I$ and write
$i\preceq j$ when $\model_i\subseteq \model_j$.  Instead of
conditioning on a single model, we then advocate the use of the
posterior distribution
\begin{equation}
  \label{eq:Q-post-true}
  Q_i(\pi_0) := P(\pi_0\,|\,\{\model:\model\subseteq \model_i\},\mathbf{Y}_n)
  = 
  \frac{\sum_{j\preceq i}
    P(\pi_0\,|\,\model_j,\mathbf{Y}_n)P(\model_j\,|\,\mathbf{Y}_n)}{ 
    \sum_{j\preceq i} P(\model_j\,|\,\mathbf{Y}_n)}
\end{equation}
obtained by conditioning on the family of all submodels of $\model_i$.
Intuitively, the proposed choice of $Q_i$ introduces the knowledge
that the data-generating distribution $\pi_0$ is in $\model_i$ all the
while capturing remaining posterior uncertainty with respect to
submodels of $\model_i$.  This does not yet escape from the problem of
circular reasoning, since (\ref{eq:Q-post-true}) involves the
posterior probabilities $P(\model_j\,|\, \mathbf{Y}_n)$ that we are
trying to approximate.  However, this problem can be overcome as we
argue in Section~\ref{subsec:singular-bic}.

For simpler notation, let $L'(\model_i):=L'_{Q_i}(\model_i)$ when
$Q_i$ is chosen according to our proposal from (\ref{eq:Q-post-true}).
We obtain from~(\ref{eq:BIC-Q}) and~(\ref{eq:Q-post-true}) that
\begin{equation}
  \label{eq:approx-ave-Qi}
  L'(\model_i) = 
  \frac{1}{\sum_{j\preceq i} P(\model_j\,|\,\mathbf{Y}_n)}\cdot 
  \sum_{j\preceq i} L'_{ij}\, P(\model_j\,|\,\mathbf{Y}_n),
\end{equation}
with
\begin{equation}
  \label{eq:L'ij(Yn)}
  L'_{ij} = P(\mathbf{Y}_n\,|\, \hat\pi_i,\model_i) \cdot
  \Lambda_{ij}(\mathbf{Y}_n)
\end{equation}
and
\begin{equation}
  \label{eq:Penalty_ij(Yn)}
  \Lambda_{ij}(\mathbf{Y}_n) =  \int_{\model_j} 
  \frac{(\log n)^{ m_i(\pi_0)-1 }}{n^{\lambda_i(\pi_0)}}
  \;dP(\pi_0\,|\,\model_j,\mathbf{Y}_n).
\end{equation}
The integral $\Lambda_{ij}(\mathbf{Y}_n)$ is the expectation of a term
measuring the complexity of model $\model_i$ under the posterior
distribution given the submodel $\model_j$.  The integration problem
in~(\ref{eq:Penalty_ij(Yn)}) may seem complicated at this point but,
in fact, for the statistical problems we have in mind, the computation
of (\ref{eq:Penalty_ij(Yn)}) is trivial because the integrand is
(almost surely) constant.  Indeed, all singular model
selection problems we know satisfy the following condition:

\begin{enumerate}
\item[] \label{generic-learn-coffs} For any $i\in I$ and
  $j\preceq i$, there are two constants $\lambda_{ij}$ and $m_{ij}$
  such that
  \begin{equation}
    \label{eq:learn-coeff-almost-sure-constant}
    \lambda_i(\pi_0) = \lambda_{ij} \quad\text{and}\quad
    m_i(\pi_0) = m_{ij}
  \end{equation}
  for all $\pi_0$ in a set $A_{ij}\subseteq\model_j$ with
   $P(A_{ij}\,|\,\model_j,\mathbf{Y}_n)=1$.
\end{enumerate}

  With such generic
values for learning coefficient and multiplicity, the integral
$\Lambda_{ij}(\mathbf{Y}_n)$ does not depend on the data $\mathbf{Y}_n$ and equals
\[
\Lambda_{ij}(\mathbf{Y}_n) = \frac{(\log n)^{ m_{ij}-1 } }{n^{\lambda_{ij}}}.
\]
In this case,
\begin{equation}
  \label{eq:L'ij}
  L'_{ij}
  = P(\mathbf{Y}_n\,|\, \hat\pi_i,\model_i)\cdot  \frac{(\log n)^{ m_{ij}-1 }
  }{n^{\lambda_{ij}}} 
\end{equation}
becomes easy to evaluate in statistical practice.

\begin{example}
  Consider again the reduced-rank regression model from
  Example~\ref{ex:red-reg}.  As mentioned, \cite{aoyagi-watanabe:2005}
  have shown that the learning coefficient $\lambda_i(\pi_0)$ and its
  multiplicity $m_i(\pi_0)$ only depend on the rank $j$ associated
  with $\pi_0$.  Hence, for any pair $j\le i$, there are constants
  $\lambda_{ij}$ and $m_{ij}$ such
  that~(\ref{eq:learn-coeff-almost-sure-constant}) holds for all
  $\pi_0$ in $\model_j\setminus \model_{j-1}$.  The exceptional set
  $\model_{j-1}$ corresponds to the matrices of rank at most $j-1$ and
  is a nullset among the matrices of rank at most $j$.  In
  Table~\ref{tab:learn-coeffs-redreg}, we listed numerical values of
  $\lambda_{ij}$ for the special case of $N=5$ responses and $M=3$
  covariates.

  The fact that~(\ref{eq:learn-coeff-almost-sure-constant}) holds for
  reduced-rank regression would also be clear if we did not know
  explicit formulas for the learning coefficients and their
  multiplicities.
  Consider model $\model_i$, and let $j\le i$.  Our claim is then that
  $\lambda_i(\pi_0)$ and $m_i(\pi_0)$ are functions of $\pi_0$ that
  are constant on a set that has probability one under
  $P(\pi_0\,|\,\model_j,\mathbf{Y}_n)$.  The pair
  $(\lambda_i(\pi_0),m_i(\pi_0))$ is determined by the asymptotics of
  a Laplace integral.  Using that $\model_i$ is a submodel of the
  regular family of \emph{all} Gaussian distributions and that
  $\model_i$ is parametrized by a polynomial map, the phase function
  of the Laplace integral can be taken to be a polynomial; compare
  e.g.~Section 2 in \cite{drton:lin:weihs:zwiernik:2014} or Lemma 1 in
  \cite{aoyagi-watanabe:2005}.  Moreover, this polynomial has
  coefficients that are polynomial functions of $\pi_0$.  When making
  this statement, we identify $\pi_0$ with the $N\times M$ matrix of
  regression coefficients.  By the theory discussed in
  \cite{Watanabe:Book}, if
  $\pi_0=\pi_0(\boldsymbol{\omega}_j)\in\model_j$ then there are
  generic values $\lambda_{ij}$ and $m_{ij}$ such that
  $(\lambda_i(\pi_0),m_i(\pi_0))\not=(\lambda_{ij},m_{ij})$ if and
  only if $\boldsymbol{\omega}_j$ satisfies a polynomial equation
  $g_{ij}(\boldsymbol{\omega}_j)=0$ that does not hold for all points
  in $\Omega_j$.  Here,
  $\Omega_j=\mathbb{R}^{N\times i}\times\mathbb{R}^{i\times M}$ is the
  parameter space of $\model_j$.  Since $g_{ij}$ is a nonzero
  polynomial, its zero set has measure zero by the lemma in
  \cite{okamoto:1973}.  Consequently,
  $(\lambda_i(\pi_0),m_i(\pi_0))=(\lambda_{ij},m_{ij})$ holds almost
  surely under $P(\pi_0\,|\,\model_j,\mathbf{Y}_n)$.
\end{example}

The reasoning just given applies verbatim to the factor analysis model
treated later, and to models for categorical data such as latent class
models or Binomial mixtures.  For mixtures of Gaussians some
additional insights are needed to arrive at a polynomial setup
but~(\ref{eq:learn-coeff-almost-sure-constant}) still holds (Section
7.8 in \citeauthor{Watanabe:Book}, \citeyear{Watanabe:Book}).  We note
that the model from Examples~\ref{ex:intro} and~\ref{ex:intro2} has
$\lambda_{21}=3/4$ and $\lambda_{22}=3/2$.  Outside the algebraic
realm, it is more difficult to make a general statement about generic
values of learning coefficients.  Nonetheless, we
expect~(\ref{eq:learn-coeff-almost-sure-constant}) to hold in all
model selection problems of practical interest; compare also Remark
1.8 in \cite{Watanabe:Book}.

\subsection{Singular BIC}
\label{subsec:singular-bic}

Even if we are able to evaluate all the integrated approximations
$L'_{ij}$ for $j\preceq i$ in the generic situation
from~(\ref{eq:L'ij}), our proposed approximation $L'(\model_i)$
remains impractical because it is a weighted average with the weights
being the posterior model probabilities
$P(\model_j\,|\, \mathbf{Y}_n)$ that
we seek to approximate in the
first place.  To make this fact more transparent, we
rewrite~(\ref{eq:approx-ave-Qi}) using that
$P(\model_j\,|\,\mathbf{Y}_n)\propto L(\model_j)P(\model_j)$, which
gives
\begin{align}
  \label{eq:L-eqn-2}
  L'(\model_i) &= \frac{1}{\sum_{j\preceq i} L(\model_j)P(\model_j)} \cdot
  \sum_{j\preceq i} L'_{ij}\,  L(\model_j)P(\model_j).
\end{align}
We see explicitly that $L'(\model_i)$, the supposed proxy to marginal
likelihood, is a function of the actual marginal likelihood
$L(\model_i)$ as well as the marginal likelihood $L(\model_j)$ of any
submodel indexed by $j\prec i$.  Of course, there would hardly
be any interest in a proxy  $L'(\model_i)$ once the marginal
likelihood $L(\model_i)$ has been computed.

This said, equation~(\ref{eq:L-eqn-2}) also leads to a way out of this
dilemma.  Observe that in~(\ref{eq:L-eqn-2}), the marginal likelihood
for model $\model_i$ appears twice, first in approximation on the left
hand side and then as an exact value on the right side when
considering summation index $j=i$.  This motivates building a
`fix-point equation system' by replacing each marginal likelihood
$L(\model_j)$ on the right hand side of~(\ref{eq:L-eqn-2}) by its
approximation $L'(\model_j)$.  We arrive at the equation system
\begin{equation}
  \label{eq:L'-system-a}
  L'(\model_i)= \frac{1}{\sum_{j\preceq i} L'(\model_j)P(\model_j)} \cdot
  \sum_{j\preceq i} L'_{ij}\,L'(\model_j)P(\model_j), \quad i\in
   I,
\end{equation}
where the $L'_{ij}$ and the $P(\model_j)$ are known constants and the
desired marginal likelihood approximations $L'(\model_i)$ are the
unknowns that we wish to solve for.  We emphasize that
(\ref{eq:L'-system-a}) is not mathematically deduced
from~(\ref{eq:L-eqn-2}), it is simply an equation system that we
heuristically motivated.  Now, if we can solve the nonlinear equation
system in~(\ref{eq:L'-system-a}) and obtain a solution with all
$L'(\model_i)>0$ then we have computed a practical approximation to
the marginal likelihood of each considered model $\model_i$, $i\in I$.

Our next observation is that the equations in~(\ref{eq:L'-system-a})
indeed have a positive solution and that this solution is unique.  To
show this, we clear denominators and consider the polynomial equation
system
\begin{equation}
  \label{eq:L'-system}
  \sum_{j\preceq i} \big[ L'(\model_i)  -
   L'_{ij}\big] L'(\model_j)P(\model_j)=0, \quad i\in
   I.
\end{equation}

\begin{proposition}
  \label{prop:positive-solutions}
  The equation system in (\ref{eq:L'-system}) has a unique solution
  with all unknowns $L'(\model_i)>0$.
\end{proposition}
\begin{proof}
  Let $i$ be any minimal element of the poset $I$.  Then $j=i$ is
  the only choice for the index $j$, and the equation from
  (\ref{eq:L'-system}) reads
  \[
  \big[ L'(\model_i)  -
   L'_{ii}\big] L'(\model_i)P(\model_i) = 0.
  \]
  With $P(\model_i)>0$, the equation has the unique positive solution
  \[
  L'(\model_i) = L'_{ii}>0,
  \]
  which coincides with the exponential of the usual BIC for model
  $\model_i$.
  
  Consider now a non-minimal index $i\in I$.  Proceeding by induction,
  assume that positive solutions $L'(\model_j)$ have been computed for
  all $j\prec i$, where $j\prec i$ if $\model_j\subsetneq \model_i$.
  Then $L'(\model_i)$ solves the quadratic equation
  \begin{equation}
    \label{eq:quadratic}
     L'(\model_i)^2 + b_i\cdot L'(\model_i) - c_i =0
  \end{equation}
  with
  \begin{align}
    \label{eq:bi}
    b_i &= -L'_{ii} + \sum_{j\prec i}
    L'(\model_j)\frac{P(\model_j)}{P(\model_i)},
    \\
    \label{eq:ci}
    c_i &=  \sum_{j\prec i} L'_{ij}\, L'(\model_j)\frac{P(\model_j)}{P(\model_i)}.
  \end{align}
  Since $c_i>0$ by the induction hypothesis, 
  (\ref{eq:quadratic}) has the unique positive solution
  \begin{equation}
    \label{eq:sbic:sol}
  L'(\model_i) = \frac{1}{2}\left(-b_i+\sqrt{b_i^2+4c_i} \right).
  \qedhere
  \end{equation}
\end{proof}

Based on Proposition~\ref{prop:positive-solutions}, we make the
following definition in which we consider the equation system
from~(\ref{eq:L'-system}) under the default of a uniform prior on
models, that is, $P(\model_i)=1/|I|$ for $i\in I$.

\begin{definition}
  \label{def:sbic}
  The \emph{singular Bayesian information criterion} for model
  $\model_i$ is
  \[
  \SBIC(\model_i) = \log L'(\model_i),
  \]
  where $(L'(\model_i):i\in I)$ is the unique solution
  to the equation system
  \begin{equation}
    \label{eq:L'-system-uniform}
    \sum_{j\preceq i} \big[ L'(\model_i)  -
    L'_{ij}\big] L'(\model_j)=0, \quad i\in
    I,
  \end{equation}
  that has all components positive.
\end{definition}

According to~(\ref{eq:L-eqn-2}), $\SBIC(\model_i)$ is the logarithm of
a weighted average of the approximations $L_{ij}'$, with the weights
depending on the data.  As discussed in Section~\ref{sec:background},
for priors with smooth and positive densities it holds that
$\lambda_i(\pi_0)\le \dim(\model_i)/2$ and $m_i(\pi_0)\ge 1$ for all
$\pi_0\in\model_i$.  Assuming $n\ge 3$, this implies that
\[
\frac{n^{\lambda_i(\pi_0)}}{ (\log n)^{m_i(\pi_0)-1}} \le n^{\dim(\model_i)/2}.
\]
Consequently, the singular BIC is of the form
\[
\SBIC(\model_i) = \log P(\mathbf{Y}_n\,|\, \hat\pi_i,\model_i) - \text{penalty}(\model_i),
\]
where $\text{penalty}(\model_i)$ is a data-dependent penalty that satisfies
\[
\text{penalty}(\model_i)\le \dim(\model_i)/2\cdot \log(n)
\]
and thus is milder than that in the usual BIC.

\begin{remark}
  \label{rem:nonuniform}
  While we envision that the use of a uniform prior on models in
  Definition~\ref{def:sbic} is reasonable for many applications,
  deviations from this default can be of interest; compare, for
  instance, \cite{Nobile:2005} who discusses priors for the number of
  components in mixture models.  Via equation
  system~(\ref{eq:L'-system}), a non-uniform prior on models can be
  readily incorporated in the definition of the singular BIC.  Later
  large-sample results would not be affected.
\end{remark}

\begin{remark}
  \label{rem:sbic-bounds}
  The sBIC defined by~(\ref{eq:L'-system-uniform}) is a function of
  the approximations $L_{ij}'$ from~(\ref{eq:L'ij}), which in turn
  depend only on the maxima of the likelihood functions and the
  numbers $\lambda_{ij}$ and $m_{ij}$.  In our treatment so far the
  $\lambda_{ij}$ are learning coefficients and the $m_{ij}$ their
  multiplicities; recall~(\ref{eq:learn-coeff-almost-sure-constant}).
  However, as we will see for applications discussed in
  Section~\ref{sec:numerical-mix}, interesting versions of sBIC also
  arise when setting the $\lambda_{ij}$ and $m_{ij}$ equal to 
  bounds on learning coefficients and multiplicities, respectively.
\end{remark}

\section{Large-sample properties}
\label{sec:large-sample}

As mentioned in the introduction, Schwarz's BIC with its
dimension-based penalty has been shown to be consistent in a number of
settings, including many singular model selection problems.
Theorem~\ref{thm:sbic-consistent} in this section asserts similar
consistency for the singular BIC from Definition~\ref{def:sbic}.  We
then proceed to show that $\SBIC$ possesses the properties we set out
to obtain.  Indeed, by Proposition~\ref{prop:true-vs-true}, the
data-dependent penalty in $\SBIC$ successfully adapts to the
data-generating distribution, meaning that in large samples the
penalty that $\SBIC$ assigns to a true model $\model_i$ agrees with
the penalty obtained from the (in practice unknown) learning
coefficient $\lambda_i(\pi_0)$ and its multiplicity $m_i(\pi_0)$.  As
stated in Theorem~\ref{thm:bayes}, it follows that the $\SBIC$ is
indeed Bayesian in the sense that it deviates from the log-marginal
likelihood by terms that are bounded in probability.

\subsection{Setup and assumptions}

We consider a finite set of models $\{\model_i: i\in I\}$ that
is closed under intersection.
Fix a
data-generating distribution $\pi_0\in\bigcup_{i\in I} \model_i$.  A
model $\model_{i}$ is \emph{true} if $\pi_0\in\model_{i}$.
Otherwise, $\model_i$ is \emph{false}.  Since the set of model is
closed under intersection, there is a unique \emph{smallest true}
model, which we denote by $\model_{i_0}$ for index $i_0\in I$.

Throughout this section, we assume that Watanabe's result
from~(\ref{eq:wata-pi0}) holds with generic learning coefficients
$\lambda_{ij}$ and multiplicities $m_{ij}$ as
in~(\ref{eq:learn-coeff-almost-sure-constant}).  Then
$\SBIC(\model_i)$ is computed from the approximations in~(\ref{eq:L'ij}),
where
\begin{equation}
  \label{eq:Bayes-complexity-factor}
  \frac{n^{\lambda_{ij}}}{(\log n)^{m_{ij}-1}},
\end{equation}
acts as a measure of complexity of model $\model_i$.  We refer to this
measure of complexity as the \emph{(generic) Bayes complexity} of
$\model_i$ along its submodel $\model_j$.  Let $\le$ denote the
lexicographic order on $\mathbb{R}^2$, that is, $(x_1,y_1)\le
(x_2,y_2)$ if $x_1<x_2$ or if $x_1=x_2$ and $y_1\le y_2$.  Then two Bayes
complexities are ordered as
\[
\frac{n^{\lambda_{1}}}{(\log n)^{m_{1}-1}} \le \frac{n^{\lambda_{2}}}{(\log
  n)^{m_{2}-1}} 
\]
for all large $n$ if and only if $(\lambda_1,-m_1)\le (\lambda_2,-m_2)$.

\smallskip In order to present a general result, we make the following
assumptions about the behavior of likelihood ratios and the learning
coefficients and their multiplicities, under a fixed data-generating
distribution $\pi_0$:

\begin{enumerate}[label=(A\arabic*)]
\item \label{A1} For any two true models $\model_i$ and $\model_k$,
  the sequence of likelihood ratios
  \[
  \frac{P(\mathbf{Y}_n\,|\,\hat\pi_k,\model_k)}{P(\mathbf{Y}_n\,|\,\hat\pi_i,\model_i)}
  \]
  is bounded in probability 
  as $n\to\infty$.
\item \label{A2} For any pair of a true model $\model_i$ and a false
  model $\model_k$, there is a constant $\delta_{ik}>0$ such that with
  probability tending to 1 as $n\to\infty$, we have that
  \[
  \frac{P(\mathbf{Y}_n\,|\,\hat\pi_k,\model_k)}{P(\mathbf{Y}_n\,|\,\hat\pi_i,\model_i)}
  \le e^{-\delta_{ik} n}.
  \]
\item \label{A3} The generic Bayes complexities are increasing with
  model size in the sense that for any model indices $i,k\in I$ and
  submodel indices $j,l\in I$, we have that
  \begin{align*}
       \qquad\qquad 
       (\lambda_{ij},-m_{ij}) &<
    (\lambda_{kj},-m_{kj}) &&\text{if}\quad j\preceq i\prec k,\quad
    \text{and}\qquad\qquad\\
    (\lambda_{il},-m_{il}) &<
    (\lambda_{ij},-m_{ij}) &&\text{if}\quad l\prec j\preceq i.
  \end{align*}
\end{enumerate}

The reader is accustomed with assumptions \ref{A1} and \ref{A2} from
any treatment of consistency of Schwarz's BIC.  Assumption \ref{A1},
the more subtle of the two conditions, holds in problems that involve
possibly singular submodels of exponential families and other
well-behaved models.  In such problems, the likelihood ratios in
\ref{A1} typically converge to a limiting distribution
\citep{drton:2009}.  Examples are Gaussian models such as reduced-rank
regression and factor analysis, but also latent class and other models
for categorical data.  As also mentioned when discussing
equation~(\ref{eq:wata}), the sequence of likelihood ratios for
mixture models is typically bounded in probability when the parameter
space is compact; without compactness the sequence need not be
bounded.  For Gaussian mixtures, for instance, the log-likelihood
ratios could be of the same $\log\log (n)$ order that the
multiplicities $m_i(\pi_0)$ have an effect on
\citep{hartigan:1985,bickel:1993}.

The first set of inequalities in assumption \ref{A3} pertains to a
fixed (generic) data-generating distribution in $\model_j$ and makes
the natural requirement that among any two true models $\model_i$ and
$\model_k$, the larger model, which is taken to be $\model_k$, has the
larger Bayes complexity.  The second set of inequalities in \ref{A3}
requires that the Bayes complexity of a fixed model $\model_i$
decreases when the data-generating distribution is moved from a
generic member of a submodel $\model_j$ to a generic member of
$\model_l\subsetneq\model_j$.  Indeed, the parameters of singular
models are typically `less identifiable' at special distributions that
correspond to smaller submodels, and the second set of inequalities
quantifies such a property.  The inequalities from \ref{A3} hold in
all the aforementioned examples for which learning coefficients have
been computed; in particular, the assumption holds for the
applications we will treat later including reduced-rank regression
from Example~\ref{ex:red-reg}.

\subsection{Consistency}

Our first result clarifies that the singular BIC selects the smallest
true model in the large sample limit.  We emphasize that we fix a
data-generating distribution $\pi_0$ and then consider large-sample
limits.

\begin{theorem}
  \label{thm:sbic-consistent}
  Let $\model_{i_0}$ be the smallest true model, and let
  $\model_{\hat\imath}$ be the model selected by maximizing the
  singular BIC, that is,
  \[
  \hat\imath = \arg\max_{i\in I} \;\SBIC(\model_i).
  \]
  Under assumptions \ref{A1}-\ref{A3}, the probability that
  $\hat\imath=i_0$  tends to 1
  as $n\to\infty$. 
\end{theorem}

\begin{remark}
  The consistency result in Theorem~\ref{thm:sbic-consistent} does
  not rely on the $\lambda_{ij}$ being learning coefficients.  Indeed,
  consistency holds for any version of sBIC that is based on numbers
  $\lambda_{ij}$ and $m_{ij}$ that satisfy assumption~\ref{A3}.  We
  will explore this in the applications in
  Section~\ref{sec:numerical-mix}, where $\lambda_{ij}$ and $m_{ij}$
  will be bounds on learning coefficients and their multiplicities,
  respectively; recall also Remark~\ref{rem:sbic-bounds}.
\end{remark}

Since we are concerned with a finite set of models
$\{\model_i: i\in I\}$, the consistency result in
Theorem~\ref{thm:sbic-consistent} can be established by pairwise
comparisons.  More precisely, it suffices to show that (i) the
singular BIC of any true model is asymptotically larger than that of
any false model, and (ii) the singular BIC of a true model can be
asymptotically maximal only if the model is the smallest true model.
The comparisons (i) and (ii) are addressed in
Propositions~\ref{prop:false} and~\ref{prop:true-vs-true},
respectively.  Throughout, $(L'(\model_i):i\in I)$ refers to the
unique positive solution of~(\ref{eq:L'-system-uniform}), that is, 
$\log L'(\model_i)=\SBIC(\model_i)$.

\begin{proposition}
  \label{prop:false}
  Under assumption \ref{A2}, if model $\model_i$ is true and model
  $\model_k$ is false, then the probability that
  $\SBIC(\model_i)>\SBIC(\model_k)$ tends to 1 as
  $n\to\infty$.
\end{proposition}
\begin{proof}
  Fix an index $j\preceq i$ and a second index $l\preceq k$.  Since
  $\model_k$ is false, \ref{A2} implies that the ratio
  ${L_{kl}'}/{L_{ij}'}$ converges to zero in probability as
  $n\to\infty$, i.e., ${L_{kl}'}=o_p(L_{ij}')$.  Since $j$ was
  arbitrary, ${L_{kl}'}=o_p(L_{i\min}')$, where
  \[
  L_{i\min}'=\min\{
  L_{ij}': j\preceq i\};
  \]
  note that for fixed $i$ and varying $j$ the approximations $L_{ij}'$
  share the likelihood term and differ only in the learning
  coefficients or their multiplicities.

  According to~(\ref{eq:L'-system-a}), $L'(\model_k)$ is a weighted
  average of the terms $L_{kl}'$ with $l\preceq k$.  We obtain that
  \begin{equation}
    \label{eq:average-smaller-than-sum}
    L'(\model_k)\,\le\, \max\{ L_{kl}': l\preceq k\} \,=\, o_p(L_{i\min}') .
  \end{equation}
  Similarly, $L'(\model_i)$ is a weighted average of the $L_{ij}'$,
  $j\preceq i$, and it thus holds that
  \begin{equation}
    \label{eq:greater}
    L'(\model_i)\ge L_{i\min}'>0.
  \end{equation}
  We conclude that
  \begin{equation}
    \label{eq:false:small:op}
    L'(\model_k)=o_p(L'(\model_i)).
  \end{equation}
  It follows that $L'(\model_i)>L'(\model_k)$ with probability tending
  to 1 as $n\to\infty$, which yields the claim because
  $\SBIC(\model_i)=\log L'(\model_i)$.
\end{proof}

\begin{proposition}
  \label{prop:true-vs-true}
  Let $\model_i$ be a true model.  Then under assumptions \ref{A1}-\ref{A3},
  \[
  \SBIC(\model_i)=\log(L'_{ii_0}) +o_p(1),
  \] 
  and thus for all $i\succ i_0$, with probability tending to 1 as $n\to\infty$,
  \[
  \SBIC(\model_i)<\SBIC(\model_{i_0}).
  \]
\end{proposition}
\begin{proof}
  First note that under assumption \ref{A3} the second assertion is a
  straightforward consequence of the first;
  compare~(\ref{eq:larger-model-larger-complexity}) below. By
  exponentiating, the first assertion is seen to be equivalent to
  \[
  L'(\model_i) = L'_{ii_0}(1+o_p(1)), \quad i\succeq i_0.
  \]
  We will argue by induction on $i$.  

  To establish the base for the induction, consider the smallest
  true model, that is, $i=i_0$.  Let $j\prec i_0$.  Then we know
  from~(\ref{eq:false:small:op}) that
  $L'(\model_j)=o_p(L'(\model_{i_0}))$.  Using the exponentially fast
  decay of the ratio in~\ref{A2}, the arguments in the proof of
  Proposition~\ref{prop:false} also yield that
  $L'(\model_j)f(n)=o_p(L'(\model_{i_0}))$ for any polynomial $f(n)$.
  Since $L_{i_0j}'/L_{i_0\min}'$ is a deterministic function that
  grows at most polynomially with $n$, and since $L_{i_0\min}'\le
  L'(\model_{i_0})$ according to~(\ref{eq:greater}), we have
  \begin{equation}
    \label{eq:b-c-small}
    L_{i_0j}' L'(\model_j)=o_p(L'(\model_{i_0})^2).
  \end{equation}
  Applying these observations to the coefficients $b_{i_0}$ and
  $c_{i_0}$ from~(\ref{eq:bi}) and~(\ref{eq:ci}), we obtain that
  $c_{i_0}=o_p(L'(\model_{i_0})^2)$ and
  $b_{i_0}+L_{i_0i_0}'=o_p(L'(\model_{i_0}))$.  From the quadratic 
  equation defining $L'(\model_{i_0})$, we deduce that
  \begin{equation}
    \label{eq:quadratic-ignore-false}
    L'(\model_{i_0})^2-L_{i_0i_0}' \cdot L'(\model_{i_0}) = 
    o_p(L'(\model_{i_0})^2).
  \end{equation}
  Hence, the equation's positive solution satisfies 
  our claim, namely,
  \begin{equation}
    \label{eq:smallest-true-schwarz}
    L'(\model_{i_0}) = L_{i_0i_0}'(1+o_p(1)).
  \end{equation}

  For the induction step, assume that the claim is true for proper
  submodels of $\model_i$, that is,
  \[
  L'(\model_k) = L'_{ki_0}(1+o_p(1)), \quad i_0\preceq k \prec i.
  \]
  Further note that arguing similarly as for $i=i_0$, the
  contributions of false models to the coefficients $b_i$ and $c_i$
  from~(\ref{eq:bi}) and~(\ref{eq:ci}) are seen to negligible.  We
  thus have
  \begin{align*}
    b_i &= -L'_{ii} + \Bigg[\sum_{i_0\preceq j\prec i} L'(\model_j)\Bigg](1+o_p(1))
    \;=\; -L'_{ii} + \Bigg[\sum_{i_0\preceq j\prec i}
      L'_{ji_0}\Bigg](1+o_p(1))
  \end{align*}
  and 
  \begin{align*}
    c_i     &= \Bigg[\sum_{i_0\preceq j\prec i} L'_{ij}
      L'(\model_j)\Bigg](1+o_p(1))\;
    =\; \Bigg[\sum_{i_0\preceq j\prec i} L'_{ij}
      L'_{ji_0}\Bigg](1+o_p(1)).
  \end{align*}
  By assumptions~\ref{A1} and \ref{A3}, 
  \begin{equation}
    \label{eq:larger-model-larger-complexity}
  L'_{ki_0} = o_p(L'_{i_0i_0}), \qquad i_0\preceq k \prec i,
  \end{equation}
  and also
  \[
  L'_{ij} = o_p(L'_{ii_0}), \qquad i_0\preceq j \preceq i.
  \]
  We obtain that
  \begin{align*}
    b_i &= -L'_{ii} + L'_{i_0i_0}(1+o_p(1)) = L'_{i_0i_0}(1+o_p(1))
  \end{align*}
  and
  \begin{align*}
    c_i &= L'_{ii_0} L'_{i_0i_0}(1+o_p(1)).
  \end{align*}
  Consequently,
  \begin{align*}
    L'(\model_i) &= \frac{1}{2} \left( -b_i +\sqrt{ b_i^2+4c_i}\right)
    \\
    &= \frac{1}{2} \left(-L'_{i_0i_0}+\sqrt{(L'_{i_0i_0})^2+4L'_{ii_0}
        L'_{i_0i_0}}\;\right)(1+o_p(1))\\
    &= \frac{1}{2}
    \left(-L'_{i_0i_0}+\sqrt{(L'_{i_0i_0})^2+4L'_{ii_0} 
        L'_{i_0i_0} + (2L'_{ii_0})^2}\;\right)(1+o_p(1)),
  \end{align*}
  where the last equality follows from $L'_{ii_0} = o_p(L'_{i_0i_0})$.
  However, this is what was to be shown because
  \begin{equation*}
    \frac{1}{2}
    \left(-L'_{i_0i_0}+\sqrt{(L'_{i_0i_0})^2+4L'_{ii_0} 
        L'_{i_0i_0} + (2L'_{ii_0})^2}\,\right)
     = L'_{ii_0}.
     \qedhere
  \end{equation*}
\end{proof}

\begin{remark}
  While we do not pursue this here, it would be interesting to
  establish further consistency properties for sBIC.  For instance,
  one could seek to adapt the results in \cite{Gassiat:2013} to give
  strong consistency results for sBIC.  \cite{Gassiat:2013} consider
  general information criteria for order selection, that is, for
  problems in which the set of models is totally ordered by inclusion
  (as in mixture modeling or factor analysis).  No upper bound on the
  number of such models is assumed in their work.
\end{remark}

\subsection{Connection to marginal likelihood}
\label{sec:bayes}

Under assumption~\ref{A2}, the marginal likelihood of a false model is
with high probability exponentially smaller than that of any true
model.  The frequentist large-sample behavior of Bayesian model
selection procedures is thus primarily dictated by the asymptotics of
the marginal likelihood integrals of true models.

As pointed out in Section~\ref{sec:bic:sing}, the usual BIC
from~\eqref{def:bic} with penalty depending solely on model dimension
generally does not reflect the asymptotic behavior of the marginal
likelihood of a true model that is singular, which is given
by~(\ref{eq:watanabe-approx-pi0}).  Consequently, as the sample size
increases, the Bayes factor obtained by forming the ratio of the
marginal likelihood integrals for two true models may in- or decrease
at a rate that is different from the rate for an approximate Bayes
factor formed by exponentiating the difference of the two respective
BIC scores.  Hence, there is generally nothing Bayesian about the
usual BIC in singular model selection problems.  In contrast, the new
singular BIC is connected to the large-sample behavior of the
log-marginal likelihood.

\begin{theorem}
  \label{thm:bayes}
  Let $\model_i$ be a true model, let $\model_{i_0}$ be the smallest
  true model, and let $\pi_0$ be a generic distribution in
  $\model_{i_0}$.  Then under assumptions~\ref{A1}-\ref{A3}, the
  marginal likelihood of $\model_i$ satisfies
  \[
  \log L(\model_i) = \SBIC(\model_i) +O_p(1).
  \]
\end{theorem}
\begin{proof}
  By Proposition~\ref{prop:true-vs-true} and~(\ref{eq:L'ij}),
  \begin{align*}
  \SBIC(\model_i) &= \log(L'_{ii_0})+ o_p(1) \\
  &= \log P(\mathbf{Y}_n\,|\,
  \hat\pi_i,\model_i) - \lambda_{ii_0} \log(n) + \big( m_{ii_0}-1
  \big) \log\log(n) + o_p(1).
  \end{align*}
  By~(\ref{eq:learn-coeff-almost-sure-constant}),
  \[
  \lambda_i(\pi_0) = \lambda_{ii_0}, \quad m_i(\pi_0) = m_{ii_0}.
  \]
  The claim thus follows from~(\ref{eq:wata}), which in turn follows
  from Watanabe's result~(\ref{eq:wata-pi0}) and assumption \ref{A1}.
\end{proof}

\section{Applications in multivariate analysis}
\label{sec:numerical}

We apply $\SBIC$ to two singular model selection problems arising in
multivariate analysis.  First, we consider the problem of selecting
the rank of the matrix of regression coefficients in reduced-rank
regression and perform a simulation study that illustrates consistency
properties.  Second, we treat the problem of selecting the number of
factors in factor analysis and work with a well-known data set to show
how $\SBIC$ can lead to an improved assessment of model uncertainty.
For a third application of $\SBIC$ in multivariate analysis, we point
the reader to \cite{drton:lin:weihs:zwiernik:2014} who treat Gaussian
latent forest models with similar findings that for the examples we
report on here.

\subsection{Rank selection}
\label{subsec:redrank}

We take up the setting of reduced-rank regression from
Example~\ref{ex:red-reg} and \cite{aoyagi-watanabe:2005}.  We consider
a scenario with $N=10$ responses and $M=15$ covariates.  We randomly
generate an $N\times M$ matrix of regression coefficients $\pi$ of
fixed rank $5$.  More precisely, we fix the signal strength by fixing
the nonzero singular values of $\pi$ to be 1.2, 1.0, 0.8, 0.6 and 0.4.
The matrix $\pi$ is then obtained by drawing the left- and the
right-singular vectors according to the Haar measures on the two
relevant Stiefel manifolds.  Given $\pi$, we generate $n$ independent
and identically distributed normal random vectors according to the
reduced-rank regression model, as specified in
Example~\ref{ex:red-reg}.  Rank estimates are then obtained by
maximizing Schwarz's BIC or the new $\SBIC$, respectively.  For each
value of $n$, we run 200 simulations with varying $\pi$.

\begin{figure}[t]
  \begin{center}
    \includegraphics[width=4.55cm]{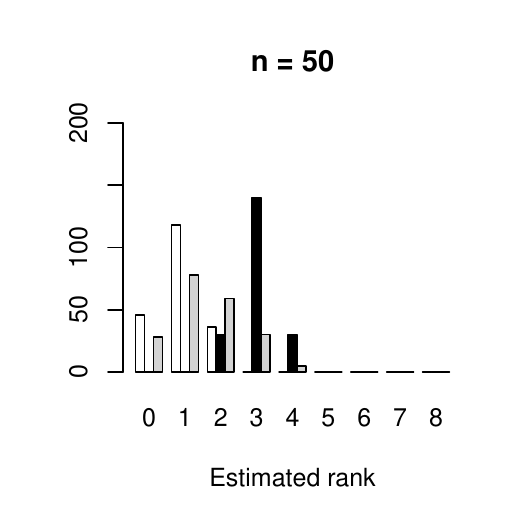}
    \includegraphics[width=4.55cm]{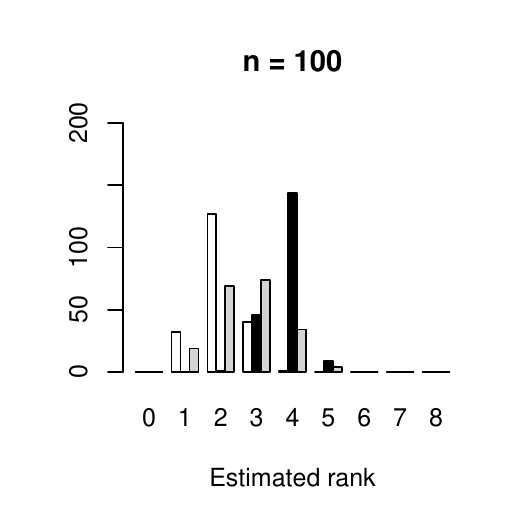}
    \includegraphics[width=4.55cm]{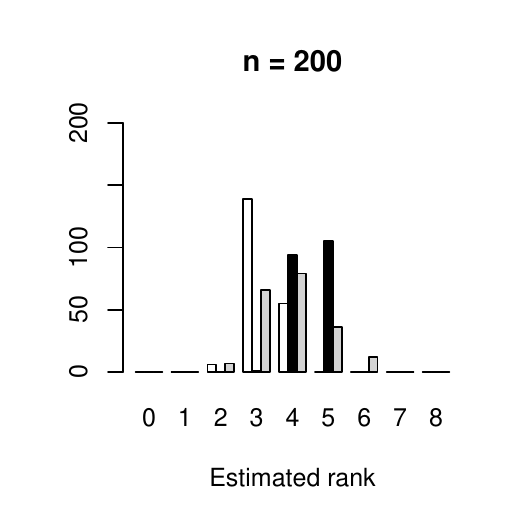}
    \includegraphics[width=4.55cm]{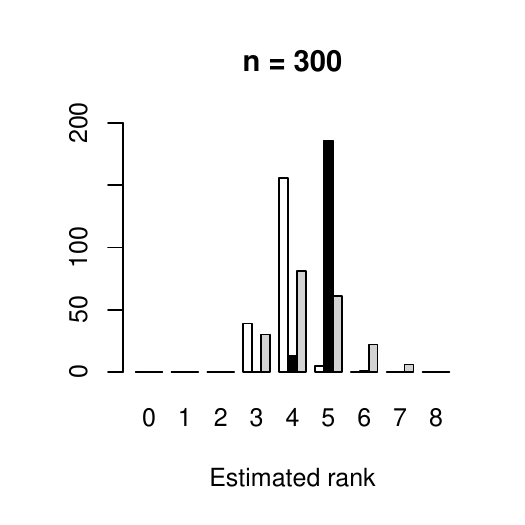}
    \includegraphics[width=4.55cm]{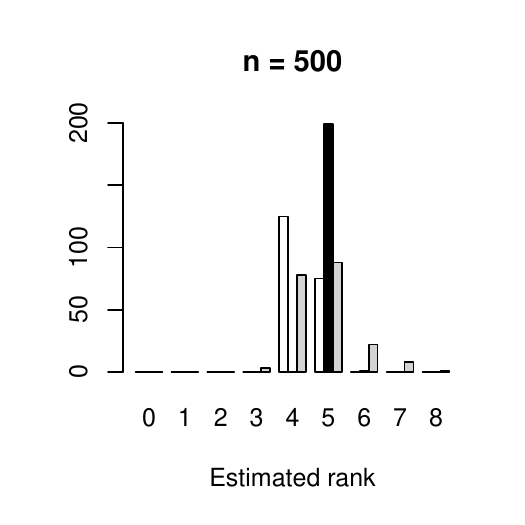}
    \includegraphics[width=4.55cm]{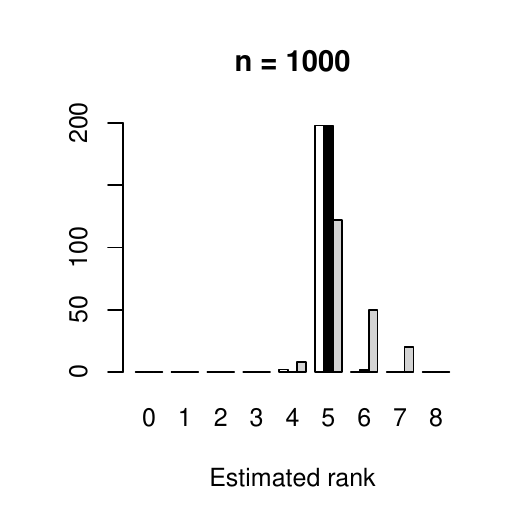}
  \end{center}\centering
  \caption{Frequencies of rank estimates in reduced-rank regression using
    Schwarz's BIC (white), WBIC (grey) and $\SBIC$ (black). Results from
    200 simulations with $10\times 15$ matrices of true rank
    5.}
  \label{fig:redreg-simulation}
\end{figure}

In our simulations, we also consider the Widely Applicable Bayesian
Information Criterion (WBIC) of \cite{watanabe:2012}.  The point of
departure in the derivation of this criterion is the fact that the
marginal likelihood can be computed by thermodynamic integration; see
also \cite{friel:2008}.  Watanabe then analyzes the large-sample
properties of the mean value obtained by applying the mean value
theorem to the thermodynamic integral.  The analysis shows that for
many models and large enough sample size $n$, the temperature at which
the mean value arises can be approximated by $\log(n)$.  We computed
WBIC for reduced-rank regression using a Metropolis-Hastings sampler
for which we adapt the computer code available on Sumio Watanabe's
website\footnote{\tt
  http://watanabe-www.math.dis.titech.ac.jp/users/swatanab/wbic2012e.html}.

We would like to stress that WBIC is not a direct competitor to our
$\SBIC$.  WBIC does not use/require knowledge of the learning
coefficients, and its computation involves integration as opposed to
the maximization in $\SBIC$.  Another important difference is that
WBIC involves an explicit choice of a prior on model parameters, where
as $\SBIC$ depends on the prior only through learning
coefficients.  The prior distribution in the code we use for WBIC has
the entries of the two matrices $\boldsymbol{\omega}_{i1}$ and
$\boldsymbol{\omega}_{i2}$ i.i.d~normal with mean zero and standard
deviation 10.  We tuned the standard deviations for the normal
distributions used for proposals in a random walk to 0.015.  Running
the sampler for 10,000 steps after 1,000 steps of burn-in gave average
acceptance rates that remained in the range from 0.1 to 0.9.

The results of the simulations are shown in
Figure~\ref{fig:redreg-simulation}, in which the new $\SBIC$ is seen
to have good rank selection properties in finite samples.  For
instance, for a sample size of $n = 300$, $\SBIC$ identifies the true
rank 5 in the vast majority of cases whereas the usual BIC selects a
rank of 3 or 4 in virtually all cases.
At $n=1000$, BIC and $\SBIC$ are perfect, with the exception of two
cases in which $\SBIC$ selects rank 6 and two cases in which BIC
selects rank 4.  The behavior of the implemented version of WBIC is
somewhat different with the ranks selected having greater variance.

Our main conclusion is that $\SBIC$ yields an improvement over the
standard dimension-based BIC in terms of frequentist rank selection
properties.  In this simulation study, $\SBIC$ also performs well
compared to WBIC but the rank selection properties of WBIC could
certainly be improved by tuning the involved prior distributions to
the problem at hand, as opposed to employing the defaults from the
computer code we applied.  Our conclusion from the comparison to WBIC
is simply that $\SBIC$ can achieve state-of-the-art performance in
rank selection.

\subsection{Factor analysis}
\label{subsec:factanal}

\citet[\S6.3]{Lopes:2004} fit factor analysis models to data
$\mathbf{Y}_n$ concerning changes in the exchange rates of 6
currencies relative to the British pound.  The sample size is $n=143$.
We write $\model_i$ for the factor analysis model with $i$ factors,
which in this example comprises multivariate normal distributions for
a random vector taking values in $\mathbb{R}^6$.  The distributions in
$\model_i$ have an arbitrary mean vector but their covariance matrix
is constrained to be of the form
$\boldsymbol{\Sigma}+\boldsymbol{\beta}\boldsymbol{\beta}'$, where
$\boldsymbol{\Sigma}$ is a diagonal matrix with positive entries and
$\boldsymbol{\beta}$ is a real $6\times i$ matrix.  This particular
covariance structure arises from conditional independence of the 6
observed random variables given $i$ latent factors.

\cite{Lopes:2004} restrict the number of factors $i$ to at most 3, so
as to not overparametrize the $6\times 6$ covariance matrix.  Their
Tables 3 and 5 report the following two sets of posterior model
probabilities obtained from Markov chain Monte Carlo computation:
  \begin{align}
    \label{eq:lopeswest:rjmcmc1}
    P(\model_1\,|\,\mathbf{Y}_n)&=0.00,& P(\model_2\,|\,\mathbf{Y}_n)&=0.88,&
    P(\model_3\,|\,\mathbf{Y}_n)&= 0.12
  \end{align}
  and 
  \begin{align}
    \label{eq:lopeswest:rjmcmc2}
    P(\model_1\,|\,\mathbf{Y}_n)&=0.00,& P(\model_2\,|\,\mathbf{Y}_n)&=0.98,&
    P(\model_3\,|\,\mathbf{Y}_n)&= 0.02.
  \end{align}
  They are based on slightly different priors for the parameters
  $(\boldsymbol{\Sigma},\boldsymbol{\beta})$ of each model.  Both
  types of priors have all parameters independent and use inverse
  Gamma distributions for the diagonal entries of
  $\boldsymbol{\Sigma}$.  The entries of $\boldsymbol{\beta}$ are
  i.i.d.~normal, but in doing so different identifiability constraints
  are used for~(\ref{eq:lopeswest:rjmcmc1})
  versus~(\ref{eq:lopeswest:rjmcmc2}).  The detailed specification of
  the prior is given in Sections 2.3 and 6.3 of \cite{Lopes:2004}.
  
  We consider these same data and compute Schwarz's BIC as well as our
  singular BIC.  We find it natural to also consider the model $\model_0$
  that postulates independence of the 6 considered changes in exchange
  rates.  Based on ongoing work of the first author and collaborators,
  we use the following learning coefficients $\lambda_{ij}$ for
  $\SBIC$:
  \begin{equation*}
    \renewcommand\arraystretch{1.2}
    \begin{array}{rcccc}
      & j=0 & j=1 & j=2 & j=3\\
      \hline
      i=0 & 3\\
      i=1 & \frac{9}{2} & 6           & \\
      i=2 & 6 & \frac{29}{4} & \frac{17}{2}  \\
      i=3 & \frac{15}{2}  & \frac{17}{2}   &   \frac{19}{2}      & \frac{21}{2}
    \end{array}
  \end{equation*}
  with all multiplicities $m_{ij}=1$.  These learning coefficients do
  not include the contribution of $6/2=3$ from the means of the
  six variables.  Note that the `top coefficient' $\lambda_{ii}$
  equals the dimension of the set of covariance matrices in model
  $\model_i$; for a computation of this dimension see, e.g.,~Theorem 2 in
  \cite{drton:2007}.

\begin{figure}[t]
  \begin{center}
    \includegraphics[scale=0.63]{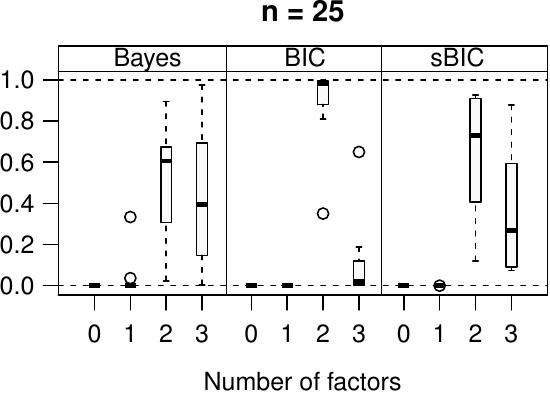} \hspace{.5cm}
    \includegraphics[scale=0.63]{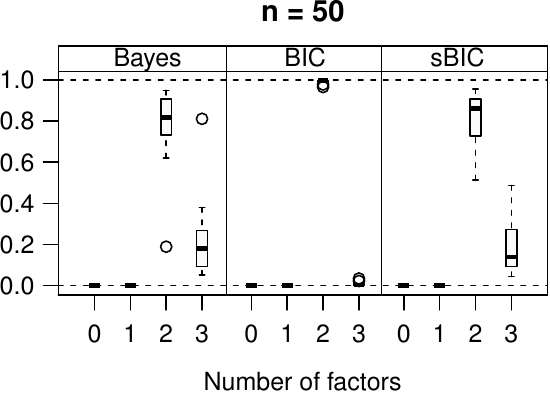}\\[0.25cm]
    \includegraphics[scale=0.63]{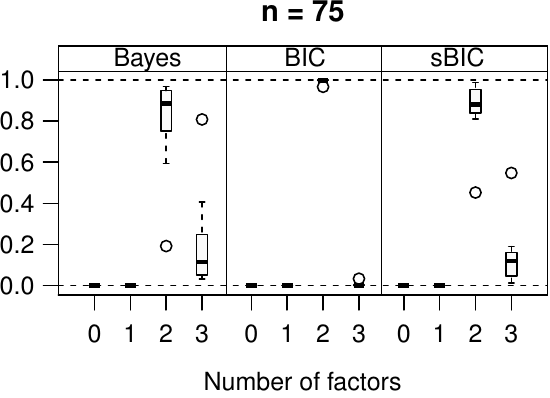} \hspace{.5cm}
    \includegraphics[scale=0.63]{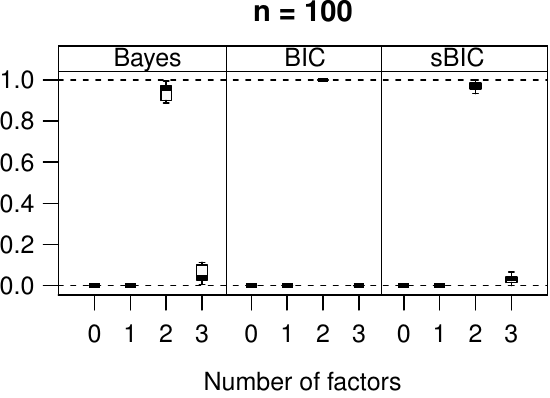}
  \end{center}\centering
  \caption{Boxplots of posterior
    model probabilities in a factor analysis of exchange rate data
    under subsampling to size $n\in\{25,50,75,100\}$: Results from a
    Markov chain Monte Carlo algorithm (`Bayes'), Schwarz's BIC and the
    new $\SBIC$.}
  \label{fig:factanal-boxplots}
\end{figure}

  Exponentiating and renormalizing either
  set of BIC scores, we obtain the following approximate posterior
  model probabilities:
  \begin{equation}
    \label{eq:lopeswest:bic}
  \begin{array}{lcccc}
     & P(\model_0\,|\,\mathbf{Y}_n) & P(\model_1\,|\,\mathbf{Y}_n) & P(\model_2\,|\,\mathbf{Y}_n)
    & P(\model_3\,|\,\mathbf{Y}_n)\\ 
    \hline
    \mathrm{BIC} & 0.0000 & 0.0000 & 0.9999 & 0.0001\\
    \SBIC & 0.0000 & 0.0000 & 0.9797 & 0.0203\\
  \end{array}
  \end{equation}
  Comparing~(\ref{eq:lopeswest:bic}) to (\ref{eq:lopeswest:rjmcmc1})
  and (\ref{eq:lopeswest:rjmcmc2}), we see that the approximation
  given by $\SBIC$ gives results that are closer to the Monte Carlo
  approximations than those from the standard BIC which leads to
  overconfidence in model $\model_2$.  Of course, this assessment is
  necessarily subjective as it pertains to a comparison with two
  particular priors $P(\pi_i\,|\,\model_i)$ in each model.
  
  To further explore the connection between the information criteria
  and fully Bayesian procedures, we subsampled the considered
  exchange rate data to create 10 data sets for each sample size
  $n\in\{25,50,75,100\}$.  For each data set we ran the Markov chain
  Monte Carlo algorithms of \cite{Lopes:2004}, focusing on the prior
  underlying~(\ref{eq:lopeswest:rjmcmc2}).  In
  Figure~\ref{fig:factanal-boxplots} we present boxplots of the four
  posterior model probabilities.  When comparing the spread in the
  approximate posterior probabilities, $\SBIC$ gives a far better
  agreement with the fully Bayesian procedure than the standard BIC.

\begin{figure}[t]
  \begin{center}
    \includegraphics[scale=0.63]{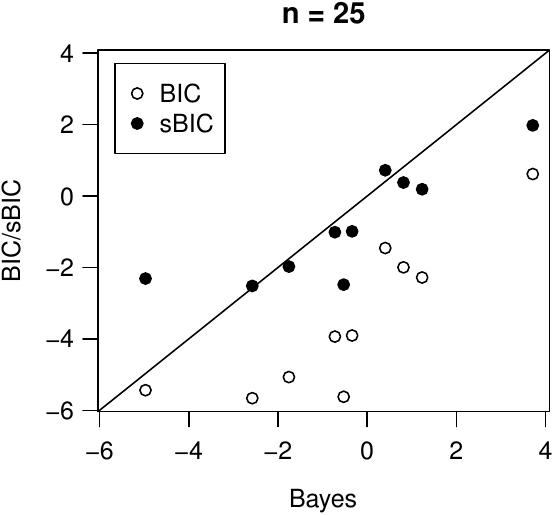} \hspace{.5cm}
    \includegraphics[scale=0.63]{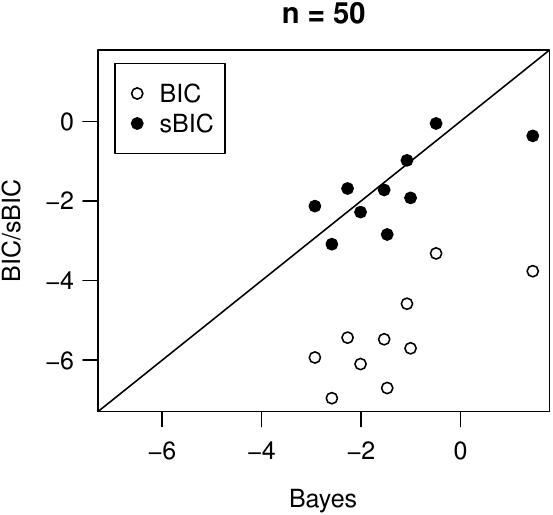}  \\[0.25cm]
    \includegraphics[scale=0.63]{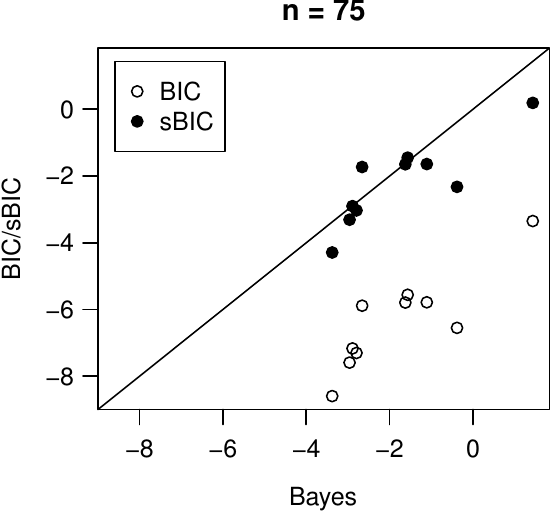} \hspace{.5cm}
    \includegraphics[scale=0.63]{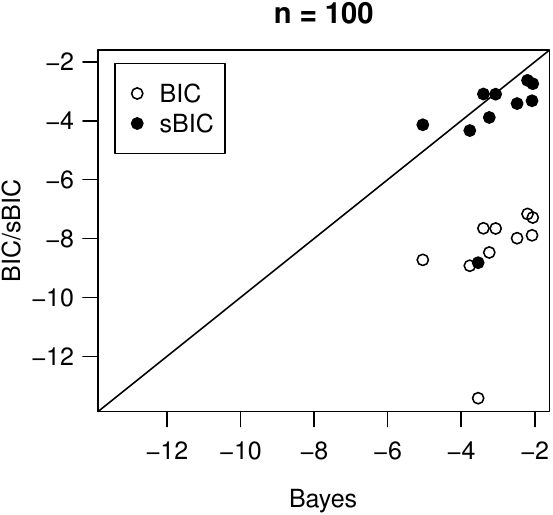}
  \end{center}\centering
  \caption{Scatter plot of log-Bayes factors comparing the results of
    a Markov chain Monte Carlo algorithm to BIC and $\SBIC$ in a
    factor analysis of exchange rate data under subsampling to size
    $n\in\{25,50,75,100\}$.} 
  \label{fig:factanal-bf}
\end{figure}
 
  For the considered data, the model uncertainty mostly concerns the
  decision between two and three factors and can be summarized by the
  Bayes factor for this model comparison.  In
  Figure~\ref{fig:factanal-bf}, we plot the log-Bayes factors obtained
  from the Markov chain Monte Carlo procedure against those computed
  via the information criteria.  The results from $\SBIC$ are seen to
  be rather close to Bayesian; the filled points in the scatter plot
  cluster around the 45 degree line.  The plot also illustrates one
  more time that BIC is overly certain about the number of factors
  being two.

\section{Applications   in mixture modeling}
\label{sec:numerical-mix}

We now apply $\SBIC$ to select the number of mixture components for
finite mixture models, which is a problem where the standard
dimension-based BIC has a tendancy to underselect the number of
components \citep[][Section 4.2]{charnigo:2007}.  Determining the
learning coefficients for mixture models can be a complicated problem
but it is possible to give simple and general bounds, and we
demonstrate that these bounds yield useful versions of $\SBIC$ (recall
Remark~\ref{rem:sbic-bounds}).  We begin with simulations for mixtures
of Binomial distributions.  Next, we fit Gaussian mixture models to
the all too familiar galaxies data \citep[e.g.][]{Roeder:1997} in
order to illustrate that $\SBIC$ allows for more posterior mass to be
assigned to larger models, which seems more in line with fully
Bayesian procedures for model determination.  Finally, we present
simulations for latent class analysis, which involves mixture models
with multi-parameter component distributions.  In this setting, the
values of the learning coefficients depend in important ways on the
choice of prior distributions, which can have substantial impact on
the model selection behavior of $\SBIC$.


\subsection{Binomial mixtures}
\label{sec:binom-mix}

Suppose $Y_{n1},\dots,Y_{nn}$ are i.i.d.~counts whose distribution
$\pi$ is modeled as a mixture of Binomial distributions.  We write
$\mathcal{B}(k,\theta)$ for the Binomial distribution with sample size
parameter $k$ and success probability $\theta\in[0,1]$.  To match
previously used notation, let $i$ denote the number of mixture
components, and let model $\model_i$ comprise the distributions
\[
\pi_i(\boldsymbol{\alpha},\boldsymbol{\theta}) \;=\; \sum_{h=1}^i
\alpha_h \mathcal{B}(k,\theta_h), 
\]
where $\boldsymbol{\alpha}=(\alpha_1,\dots,\alpha_i)$ is a vector of
unknown nonnegative mixture weights that sum to one, and
$\boldsymbol{\theta}=(\theta_1,\dots,\theta_i)\in[0,1]^i$ is a vector
of unknown success probabilities.  We assume the Binomial sample size
parameter $k$ to be known.  Throughout this subsection, we assume that
each prior distribution
$P(\boldsymbol{\alpha},\boldsymbol{\theta}\,|\,\model_i)$ has a
density that is bounded away from zero on
$\Delta_{i-1}\times [0,1]^i$.

Consider now a data-generating distribution $\pi_0\in\model_i$.  The
fiber of $\pi_0$ under the parametrization of $\model_i$ is the
preimage
\begin{equation}
  \label{eq:fiber_binom_mix}
\mathcal{F}_i(\pi_0) \;=\;
\left\{ 
(\boldsymbol{\alpha},\boldsymbol{\theta})\in\Delta_{i-1}\times [0,1]^i \::\:
\pi_i(\boldsymbol{\alpha},\boldsymbol{\theta}) = \pi_0
\right\},
\end{equation}
containing all parameter vectors
$(\boldsymbol{\alpha},\boldsymbol{\theta})$ that define the same
distribution $\pi_0$.  Here, $\Delta_{i-1}$ denotes the $(i-1)$
dimensional probability simplex.  Clearly, if
$(\boldsymbol{\alpha},\boldsymbol{\theta})\in\mathcal{F}_i(\pi_0)$
then $\mathcal{F}_i(\pi_0)$ also contains any vector that is obtained
by permuting the entries of $\boldsymbol{\theta}$ and, accordingly,
those of $\boldsymbol{\alpha}$.  When $i$ is not too large with
respect to $k$, specifically, if $2i-1\le k$, then the fiber of a
distribution $\pi_0=\pi_i(\boldsymbol{\alpha},\boldsymbol{\theta})$ in
$\model_i\setminus\model_{i-1}$ contains only the $i!$ points in
the orbit of $(\boldsymbol{\alpha},\boldsymbol{\theta})$, and
\begin{equation}
  \label{eq:binmix-model-dim}
  \dim(\model_i) \;=\; 2i-1;
\end{equation}
see Proposition 4 in \cite{teicher:1963} or also Section 3.1 in
\cite{Titterington}.

When the number of mixture components $i$ is larger than needed,
however, a severe non-identifiability problem arises.  This is
discussed, for instance, in Section 1.3 of \cite{schnatter:2006}.  We
provide some details that form the basis for bounds on learning
coefficients that we will consider for a definition of $\SBIC$.

\begin{proposition}
  \label{prop:binomial-mixtures}
  Suppose $2i-1\le k$ and consider a Binomial mixture
  $\pi_0\in\model_j\setminus \model_{j-1}$ for $j<i$.  Then the fiber
  $\mathcal{F}_i(\pi_0)$ from (\ref{eq:fiber_binom_mix}) is the
  intersection of $\Delta_{i-1}\times[0,1]^i$ with a finite union of
  $(i-j)$ dimensional affine spaces.  In particular,
  $\mathcal{F}_i(\pi_0)$ has dimension $(i-j)$.
\end{proposition}
\begin{proof}
  Since $\pi_0\in\model_j\setminus\model_{j-1}$, we have
  \begin{equation}
    \label{eq:binommix-modelj}
    \pi_0 \;=\;\sum_{h=1}^j
    \alpha_{0h} \mathcal{B}(k,\theta_{0h}), 
  \end{equation}
  where the success probabilities $\theta_{01},\dots,\theta_{0j}$ are
  pairwise disjoint and the mixture weights
  $\alpha_{01},\dots,\alpha_{0j}$ are positive.  The probabilities
  $\theta_{0h}$ and $\alpha_{0h}$ in (\ref{eq:binommix-modelj}) are
  unique up to permutation.  

  We may represent $\pi_0$ as an element of $\model_i$ by setting
  $i-j$ of the mixture weights to zero, in which case $i-j$ of the
  success probabilities can be chosen arbitrarily.  More precisely, if
  we take
  \[
  \boldsymbol{\alpha}\;=\;(\alpha_{01},\dots,\alpha_{0j},0,\dots,0)
  \;\in\;\Delta_{i-1},  
  \]
  then
  $(\boldsymbol{\alpha},\boldsymbol{\theta})\in\mathcal{F}_i(\pi_0)$
  for any vector $\boldsymbol{\theta}$ with $\theta_h=\theta_{0h}$ for
  $1\le h\le j$.  Hence, the fiber $\mathcal{F}_i(\pi_0)$ contains the
  $(i-j)$ dimensional set
  \begin{equation}
    \label{eq:bin-mix-weights-zero}
    \{(\alpha_{01},\dots,\alpha_{0j},0,\dots,0)\} \times \{
    (\theta_{01},\dots,\theta_{0j})\}\times [0,1]^{i-j}
  \end{equation}
  and its orbit under the action of the symmetric group.  
  
  A second way to represent $\pi_0$ as an element of $\model_i$ is to
  choose a vector $\boldsymbol{\theta}\in[0,1]^i$ that has precisely
  $j$ distinct entries, the distinct values being
  $\theta_{01},\dots,\theta_{0j}$.  For each index
  $h\in\{1,\dots,j\}$, let $J_h$ be the set of indices
  $l\in\{1,\dots,i\}$ such that $\theta_l=\theta_{0h}$.  Then
  $J_1,\dots,J_j$ form a partition of $\{1,\dots,i\}$.  For instance,
  if
  \begin{equation}
    \label{eq:bin-mix-thetas-duplicated}
    \boldsymbol{\theta} \;=\;
    (\theta_{01},\theta_{02},..., \theta_{0(j-1)},\theta_{0j},...,\theta_{0j})
    \;\in\; [0,1]^i,
  \end{equation}
  then $J_h=\{h\}$ for all $h<j$ and $J_j=\{j,\dots,i\}$.  
  In order for $(\boldsymbol{\alpha},\boldsymbol{\theta})$ to be in the
  fiber $\mathcal{F}_i(\pi_0)$, it needs to hold that 
  \[
  \sum_{l\in J_h} \alpha_l \;=\;\alpha_{0h},  \qquad h=1,\dots,j.
  \]
  Clearly, there are now $i-j$ degrees of freedom in the choice of the
  mixture weights.  For instance, the fiber $\mathcal{F}_i(\pi_0)$
  contains the $(i-j)$ dimensional set that 
  \begin{equation}
    \label{eq:bin-mix-thetas-dup-set}
    \left\{(\alpha_{01},\dots,\alpha_{0,j-1})\right\}\times
    \left(\alpha_{0j}\Delta_{i-j}\right) \times \{ 
    (\theta_{01},\dots,\theta_{0j},\theta_{0j},\dots,\theta_{0j})\}
  \end{equation}
  and its orbit under the action of the symmetric group. 
\end{proof}

When $i=2$ and $\theta_0=\mathcal{B}(k,2/3)$ with $k\ge 3$, then the
fiber $\mathcal{F}_i(\pi_0)$ is a union of three line segments.  This
fiber is plotted in Figure~\ref{fig:binom-mix-fiber}; the two gray
lines intersect the boundary of the probability simplex $\Delta_1$,
i.e., have $\alpha=\alpha_1=0$ or $1-\alpha=\alpha_2=0$.

\begin{figure}[t]
  \centering
  \includegraphics[scale=0.7]{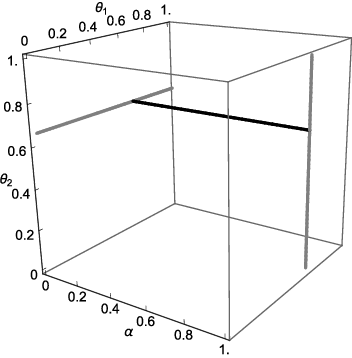}
  \caption{The fiber of a Binomial distribution in the model that
    mixes two Binomial distributions.}
  \label{fig:binom-mix-fiber}
\end{figure}

By Proposition~\ref{prop:binomial-mixtures}, the fiber
$\mathcal{F}_i(\pi_0)$ of a generic distribution $\pi_0\in\model_j$,
$j\le i$, is a set of dimension $i-j$.  The learning coefficient
$\lambda_i(\pi_0)$ for model $\model_i$ depends only on $j$ and can be
bounded by subtracting the dimension of the fiber from the model
dimension; see Section 7.3 in \cite{Watanabe:Book}.  Writing
$\lambda_{ij}=\lambda_i(\pi_0)$, we find that
\begin{equation}
  \label{eq:lambdaij-bin-mix-bound}
\lambda_{ij} \;\le\; \bar\lambda_{ij}^1 := \frac{1}{2}\left[
  \dim(\model_i) - (i-j)\right] \;=\; \frac{1}{2}\left[
  2j-1 + (i-j)\right] 
\;=\; \frac{1}{2}\left(i+j-1\right).
\end{equation}
Now it is known that the actual learning coefficient for Binomial
mixture models ($i\ge 2$) is smaller than the $\bar\lambda_{ij}^1$
from~(\ref{eq:lambdaij-bin-mix-bound}).  Indeed, for a prior density
that is bounded away from zero on $\Delta_{i-1}\times [0,1]^i$,
\cite{yamazaki-watanabe:2004} have shown that
\[
\lambda_{i,i-1}=i-\frac{5}{4},
\]
whereas $\bar\lambda_{i,i-1}^1=i-1$.  The model dimension
from~(\ref{eq:binmix-model-dim}) yields the looser bound $i-1/2$.
While no general formulas for the coefficients $\lambda_{ij}$ have
been obtained thus far, the analysis of \citet[eqns.~(5) and
(6)]{rousseach:2011} yields the tighter bound
\begin{equation}
  \label{eq:lambdaij-bin-mix-rousseau-bound}
  \lambda_{ij} \;\le\; \bar\lambda_{ij}^{0.5} := \frac{1}{2}\left[
    2j-1+ \frac{1}{2}(i-j)\right] 
  \;=\; \frac{i+3j}{4} - \frac{1}{2}.
\end{equation}
For $j=i-1$, we have $\bar\lambda_{ij}^{0.5}=\lambda_{ij}=i-5/4$
but we do not expect this to be true in general.

In light of the above discussion, we argue that using the bounds
$\bar\lambda_{ij}^{0.5}$ from~(\ref{eq:lambdaij-bin-mix-rousseau-bound})
or even the very easily derived bound $\bar\lambda_{ij}^1$
from~(\ref{eq:lambdaij-bin-mix-bound}) is more appropriate for the
definition of a Bayesian information criterion than merely working
with the model dimension from~(\ref{eq:binmix-model-dim}).  For a
numerical experiment, we generate data from a distribution $\pi_0$
that is a mixture of $4$ (but not less) Binomial distributions that
each have sample size parameter $k=30$.  Specifically, we consider the
mixture weights
\begin{align*}
  \alpha_{01} &= 1/4,
  &   \alpha_{02} &= 1/4, 
  &   \alpha_{03} &= 1/4, 
  &   \alpha_{04} &= 1/4
\end{align*}
and the success probabilities
\begin{align*}
  \theta_{01} &=1/5, 
  &   \theta_{02} &=2/5, 
  &   \theta_{03} &=3/5, 
  &   \theta_{04} &=4/5.
\end{align*}
For varying values $n$, we generate an i.i.d.~sample of size $n$ from
$\pi_0$ and select the number of mixture components by maximizing (i)
Schwarz's BIC which uses the model dimension $2i-1$, (ii)
$\overline{\SBIC}_{0.5}$ by which we mean the singular BIC computed
using the bounds $\bar\lambda_{ij}^{0.5}$
from~(\ref{eq:lambdaij-bin-mix-rousseau-bound}), and (iii)
$\overline{\SBIC}_1$ which stands for the singular BIC computed using
the $\bar\lambda_{ij}^1$ from~(\ref{eq:lambdaij-bin-mix-bound}).  Both
$\overline{\SBIC}_{0.5}$ and $\overline{\SBIC}_1$ have all
multiplicities $m_{ij}$  set to their lower bound $1$.  We repeat
the model selection 200 times.

The frequencies of how often a particular number of components was
selected by each method are depicted in
Figure~\ref{fig:binmix-simulation}, where we show plots for $n=50$,
$200$ and $500$.  The results are similar to those in the rank
selection experiment from Section~\ref{subsec:redrank} in that our
singular BIC allows one to identify the true number of components
earlier than BIC.  Both $\overline{\SBIC}_{0.5}$ and $\overline{\SBIC}_1$ alleviate some of the
overpenalization that arises when using solely the model dimension,
with $\overline{\SBIC}_{0.5}$ performing the best.

\begin{figure}[t]
  \begin{center}
    \includegraphics[width=4.55cm]{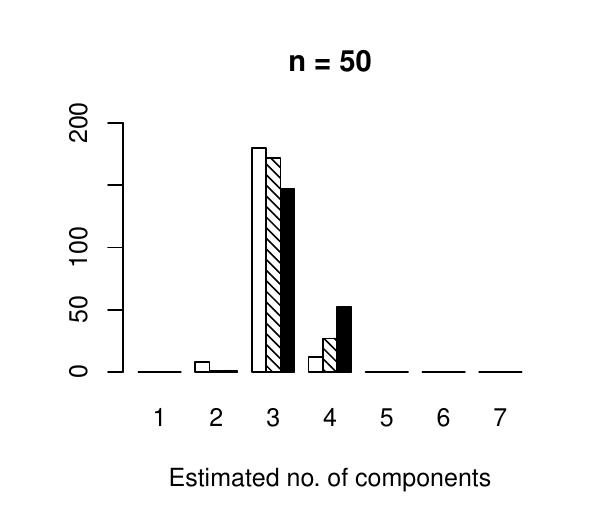}
    \includegraphics[width=4.55cm]{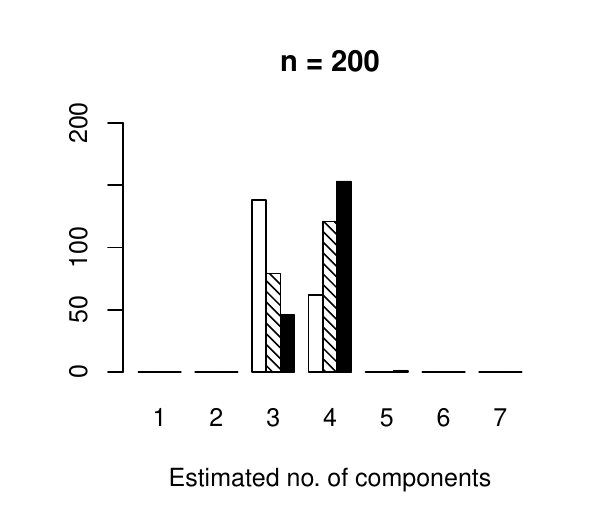}
    \includegraphics[width=4.55cm]{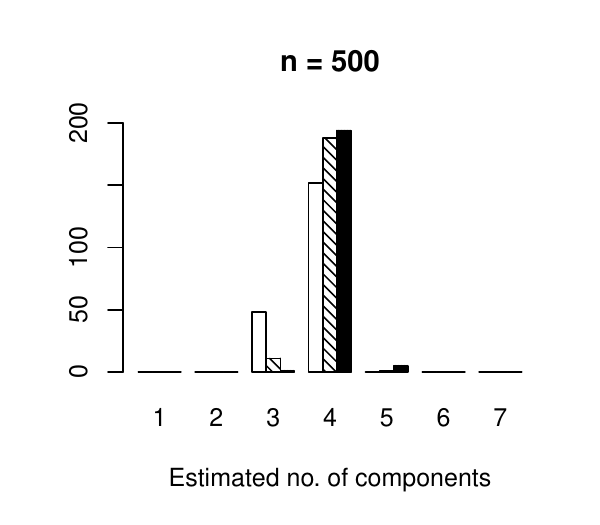}
  \end{center}\centering
  \caption{Frequencies of estimated number of Binomial mixture
    components using Schwarz's BIC (white), $\overline{\SBIC}_{0.5}$ (black), and
    $\overline{\SBIC}_1$ (striped). Results from 200 simulations with
    sample size 
    parameter $k=30$ and true number of components equal to 4.}
  \label{fig:binmix-simulation}
\end{figure}


\subsection{Gaussian mixtures}
\label{sec:gauss-mix}

\cite{aoyagi:2010} has found the learning coefficients of univariate
Gaussian mixture models when the variances of the component
distributions are known and equal to a common value.  Using them in
$\SBIC$ yields a criterion whose model selection properties are
similar to what we have shown for reduced-rank regression and Binomial
mixtures.  In this section, we report instead on a data analysis with
Gaussian mixtures where the variances are unknown and allowed to be
unequal.

Let $\model_i$ be the (univariate) Gaussian mixture model with $i$
mixture components, which comprises the distributions
\[
\pi_i(\boldsymbol{\alpha},\boldsymbol{\mu},\boldsymbol{\sigma^2}) \;=\; \sum_{h=1}^i \alpha_h \,\mathcal{N}(\mu_h,\sigma_h^2)
\]
for a vector of mixture weights
$\boldsymbol{\alpha}=(\alpha_1,\dots,\alpha_i)\in\Delta_{i-1}$,
choices of means $\boldsymbol{\mu}=(\mu_1,\dots,\mu_i)\in\mathbb{R}^i$
and variances
$\boldsymbol{\sigma^2}=(\sigma_1^2,\dots,\sigma_i^2)\in(\epsilon,\infty)^i$.
Here, we made explicit that the software we will use later, namely,
the R package {\tt mclust} \citep{fraley2012mclust}, uses a lower
bound $\epsilon>0$ to avoid the well-known singularities in the
likelihood surfaces obtained by letting one or more variances tend to
zero.  Such a lower bound also appears in
consistency theory for BIC \citep[Prop.~4.2]{keribin:2000}.

In the Gaussian mixture model $\model_i$, the fiber of a distribution
$\pi_0$ is the set
\begin{equation}
  \label{eq:fiber_gauss_mix}
  \mathcal{F}_i(\pi_0) \;=\;
  \left\{ 
    (\boldsymbol{\alpha},\boldsymbol{\mu},\boldsymbol{\sigma^2})\in\Delta_{i-1}\times
    \mathbb{R}^i\times  (\epsilon,\infty)^i\::\:
    \pi_i(\boldsymbol{\alpha},\boldsymbol{\mu},\boldsymbol{\sigma^2}) = \pi_0
  \right\}.
\end{equation}
By Proposition~1 in \cite{teicher:1963}, if
$\pi_0\in\model_i\setminus\model_{i-1}$ then $\mathcal{F}_i(\pi_0)$ is
finite with $|\mathcal{F}_i(\pi_0)|=i!$, and we have
\begin{equation}
\dim(\model_i)=3i-1.
\end{equation}
However, as discussed in Section~\ref{sec:binom-mix}, a distribution
$\pi_0\in\model_j\subset\model_i$, $j<i$, will have an infinite fiber
$\mathcal{F}_i(\pi_0)$ due to the obvious non-identifiability problem
arising from specifying the number of mixture components $i$ larger
than needed.

As described in the proof of Proposition~\ref{prop:binomial-mixtures},
a distribution $\pi_0\in\model_j\setminus \model_{j-1}$ with $j<i$ can
be represented as a member of $\model_i$ by setting $i-j$ of the
mixture weights to zero, which here leaves $i-j$ of the mean
parameters and $i-j$ of the variance parameters free.  Hence, the
fiber contains a set of dimension $2(i-j)$ that is made up of triples
$(\boldsymbol{\alpha},\boldsymbol{\mu},\boldsymbol{\sigma^2})$ that
have $\boldsymbol{\alpha}$ on the boundary of the probability simplex
$\Delta_{i-1}$.  By this fact, the learning coefficient
$\lambda_{ij}=\lambda_i(\pi_0)$ can be bounded as
\begin{equation}
  \label{eq:lambdaij-gauss-mix-bound-star}
\lambda_{ij} \;\le\; \bar\lambda_{ij}^1 := \frac{1}{2}\left[
  \dim(\model_i) - 2(i-j)\right]= \frac{1}{2}\left[
  3j-1 + (i-j)\right] = \frac{1}{2}\left(i+2j-1 \right);
\end{equation}
see again Section 7.3 in \cite{Watanabe:Book}.  For the bound to
apply, however, the density of the prior distribution
$P(\boldsymbol{\alpha},\boldsymbol{\mu},\boldsymbol{\sigma^2}\,|\,\model_i)$
has to be bounded away from zero in a neighborhood of an $2(i-j)$
dimensional subset of $\mathcal{F}_i(\pi_0)$.  This is the case if the
prior density for $\boldsymbol{\alpha}$ is bounded away from zero on
and near the boundary of the probability simplex $\Delta_{i-1}$; a
uniform distribution on $\Delta_{i-1}$ would be an example.

Keeping with $\pi_0\in\model_j\subset\model_i$, let $\Delta_{i-1}^o$
denote the interior of the probability simplex, and consider instead
the fiber
\begin{equation}
  \label{eq:fiber_gauss_mix_interior}
  \mathcal{F}_i^o(\pi_0) \;=\;
  \left\{ 
    (\boldsymbol{\alpha},\boldsymbol{\mu},\boldsymbol{\sigma^2})\in\Delta_{i-1}^o\times
    \mathbb{R}^i\times  (\epsilon,\infty)^i\::\:
    \pi_i(\boldsymbol{\alpha},\boldsymbol{\mu},\boldsymbol{\sigma^2}) = \pi_0
  \right\}
\end{equation}
that has all mixture weights nonzero.  This `positive fiber'
$\mathcal{F}_i^o(\pi_0)$ is of lower dimension than
$\mathcal{F}_i(\pi_0)$.  Indeed, equating means and variances between
mixture components in analogy to~(\ref{eq:bin-mix-thetas-duplicated})
and~(\ref{eq:bin-mix-thetas-dup-set}) shows that
$\mathcal{F}_i^o(\pi_0)$ has dimension $i-j$.  Hence, for a prior that
is supported on a subset of $\Delta_{i-1}^o$, subtraction of the fiber
dimension leads to the bound
\begin{equation}
  \label{eq:lambdaij-gauss-mix-bound}
\lambda_{ij} \;\le\; \bar\lambda_{ij}^2 := \frac{1}{2}\left[
  \dim(\model_i) - (i-j)\right] = \frac{1}{2}\left[
  3j-1 + 2(i-j)\right] = \frac{1}{2}\left(2i+ j-1\right).
\end{equation}
Nevertheless, the more refined analysis from \citet[eqns.~(5) and
(6)]{rousseach:2011} shows that the bound $\bar\lambda_{ij}^1$
from~(\ref{eq:lambdaij-gauss-mix-bound-star}) remains valid when the
prior density is bounded away from zero in a neighborhood of a point in
$\mathcal{F}_i^o(\pi_0)$.

To illustrate the above in an example, take $\pi_0=\mathcal{N}(0,1)$,
a standard normal distribution.  Then the fiber $\mathcal{F}_2(\pi_0)$
in the two component mixture model is the union of two planes and a
line intersected with
$\Delta_1\times \mathbb{R}^2\times (\epsilon,\infty)^2$.  The
structure of the fiber is as in Figure~\ref{fig:binom-mix-fiber},
except that the two gray line segment now are two-dimensional
rectangular strata.  The black part with mixture weights
$\alpha_1=\alpha$ and $\alpha_2=1-\alpha$ remains a line segment.  The
set $\mathcal{F}_i^o(\pi_0)$ then comprises only this line segment but
not the two-dimensional strata.  

\begin{figure}[t]
  \centering
    \includegraphics[scale=0.6]{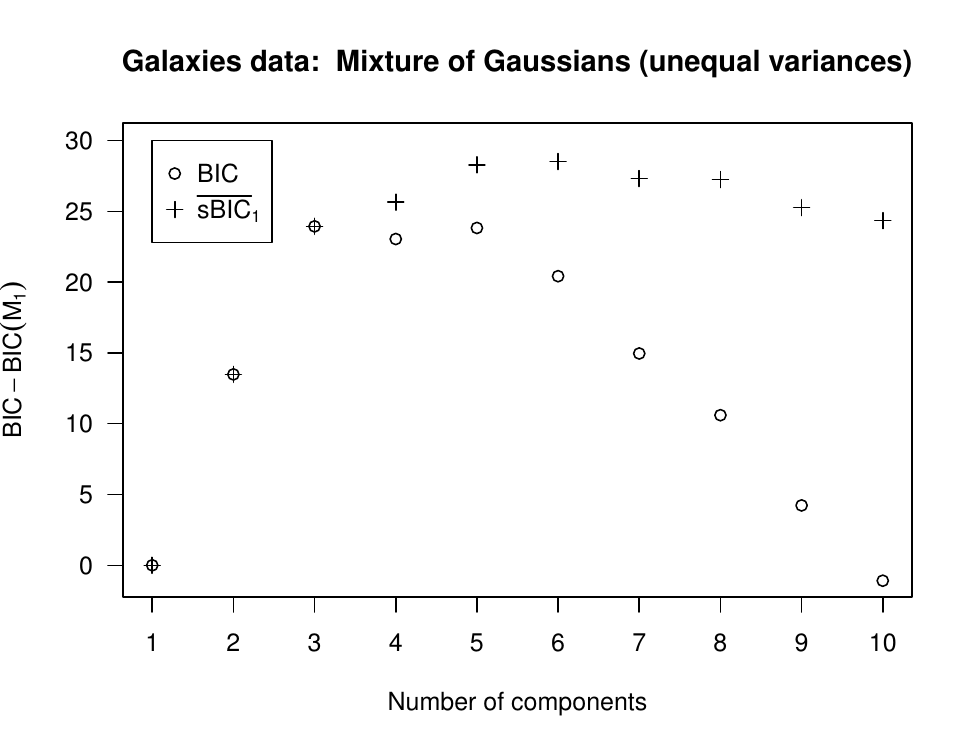}
    \vspace{-0.5cm}
  \caption{Galaxies data:  BIC and $\overline{\SBIC}_1$.} 
  \label{fig:galaxies:bic}
\end{figure}

Using the bounds $\bar\lambda_{ij}^1$
from~(\ref{eq:lambdaij-gauss-mix-bound-star}) and setting all
multiplicities to 1 yields a version of $\SBIC$, which we denote by
$\overline{\SBIC}_1$.  (We will briefly comment on the bounds
$\bar\lambda_{ij}^2$ in our conclusion.)  We apply
$\overline{\SBIC}_1$ to a familiar example, namely, the galaxies data
set discussed in depth in the review of \cite{aitkin:2001} and also in
Example 4 in \cite{Marin:2005}.  We use the EM algorithm implemented
in the R package {\tt mclust} \citep{fraley2012mclust} to fit the
mixture models and base our results on the best local maxima of the
likelihood function that were found in repeated EM runs.  For each
model, we ran the EM from 5000 random initializations that were
created by drawing, independently for each data point, a vector of
cluster membership probabilities from the uniform distribution on the
relevant probability simplex.  Figure~\ref{fig:galaxies:bic} depicts
the resulting values of BIC and $\overline{\SBIC}_1$.  These are
converted into posterior model probabilities in
Figure~\ref{fig:galaxies:post}, where we also show posterior
probabilities from the fully Bayesian analysis of
\cite{richardson:green:1997} who, in particular, adopted a uniform
prior for the mixture weights.

\begin{figure}[t]
  \centering
  \includegraphics[scale=0.6]{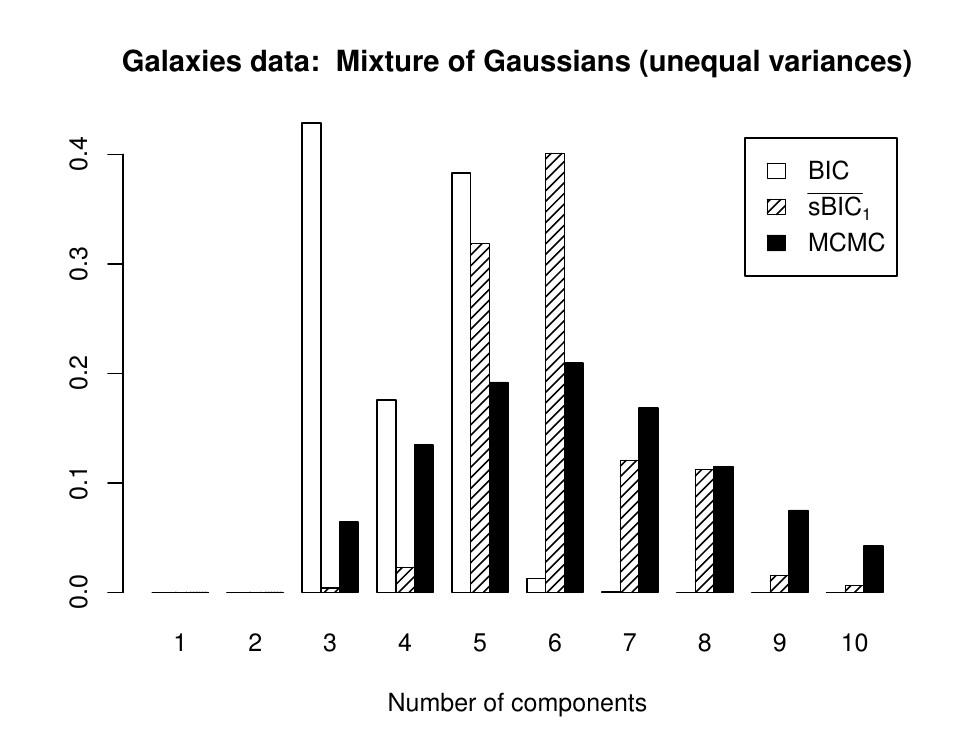}
  \vspace{-0.25cm}
  \caption{Galaxies data: Posterior model probabilities from BIC,
     $\overline{\SBIC}_1$ and MCMC per \cite{richardson:green:1997}.}
  \label{fig:galaxies:post}
\end{figure}

Figure~\ref{fig:galaxies:bic} shows that the information criteria
assign essentially the same value to the models $\model_i$ with
$i\le 3$.  This is due to poor model fit, i.e., very small maximal
likelihood under $\model_1$ and $\model_2$.  Starting with four
components differences emerge.  BIC attains high values for
$i\in\{3,4,5\}$ and decreases very quickly for larger $i$.
The decrease is nearly as quick as the increase through the models
with $i\le 3$ components, where those with $i\le 2$ seem too simple.
In contrast, $\overline{\SBIC}_1$ is largest for $i=6$ followed
closely by $i=5$, and its values remain rather large for
$i\in\{7,8\}$.  The decay for larger $i$ is far slower than for BIC.
In Figure~\ref{fig:galaxies:post}, approximate posterior model
probabilities from $\overline{\SBIC}_1$ are closer to 
the Monte Carlo estimates reported by \cite{richardson:green:1997}; see also
\cite{LeeRobert2013}.

\subsection{Latent class analysis}
\label{subsec:lca}

Our last experiments pertain to latent class analysis (LCA), in which
the joint distribution of a collection of categorical variables, the
\emph{items}, is modeled to exhibit conditional independence given a
categorical latent variable.  The values of the latent variable are
the \emph{classes}.  LCA models are also known as naive Bayes models
\citep{geiger:2001} and are related to secant varieties of Segre
varieties studied in algebraic geometry \cite[Chap.~4.1]{Drton:Book}.
We consider them here because LCA models are mixture models in which
the component distributions are taken from a family of larger
dimension.  As we will see, this makes the choice of the priors on the
mixture weights more important.

We will treat the case of $r$ binary items whose values we code to
be in $\{0,1\}$.  The model $\model_i$ with $i$ classes then
postulates that the joint probabilities for the binary items
$Y_1,\dots,Y_r$ are of the form
\[
\Pr(Y_1=y_1,\dots,Y_r=y_r) \;=\; \sum_{h=1}^i \alpha_h \prod_{l=1}^r
p_{hl}^{y_l}(1-p_{hl})^{1-y_l},
\]
where $\alpha_h$ is the probability of being in class $h$, and
$p_{hl}$ is the conditional probability of $Y_l=1$ given membership in
class $h$.  We emphasize that $r$ is equal to the dimension of the
family of distributions from which the mixture components are taken.
Counting parameters one expects that the dimension of the LCA model
$\model_i$ is
\begin{equation}
  \label{eq:lca:exp-dim}
  \min\left\{ ir+(i-1), 2^r-1\right\}.
\end{equation}
There are exceptional cases where this is not the correct dimension,
see e.g.~Example 4.1.8 in \cite{Drton:Book}.  However, Theorem 2.3 in
\cite{catalisano:2005} guarantees that all models in our below
simulation study have dimension given by~(\ref{eq:lca:exp-dim}).
All these models are also generically identifiable up to label
swapping by Corollary 5 in \cite{Allman:2009}.

Let $\pi_0\in\model_j\setminus\model_{j-1}$ for $j<i$, and assume that
$\dim(\model_i)=(i-1)+ir\le 2^r-1$.  Reasoning as in
Section~\ref{sec:gauss-mix}, dimension counting yields two simple
bounds on the learning coefficients.  For $\phi>0$,
define\footnote{Note that $r$ is the dimension of the model for a
  single mixture component.  The notation
  from~(\ref{eq:bar-lambda-phi}) matches the earlier use in
  Section~\ref{sec:binom-mix}, where $r=1$, and that in
  Section~\ref{sec:gauss-mix}, where $r=2$.}
\begin{equation}
  \label{eq:bar-lambda-phi}
  \bar\lambda_{ij}^\phi \;:=\; \frac{1}{2}
  \left[jr+j-1 + (i-j)\phi \right].  
\end{equation}
When allowing zero mixture weights $\alpha_h$, the fiber of $\pi_0$
has dimension $r(i-j)$ and the analogue
of~(\ref{eq:lambdaij-gauss-mix-bound-star}) becomes
\begin{equation}
  \label{eq:lca*-bound}
  \lambda_{ij} \;\le\;  \frac{1}{2}
  \left[\dim(\model_i) - r(i-j) \right] = \frac{1}{2}
  \left[rj+(i-1) \right] = \bar\lambda_{ij}^1 .
\end{equation}
This bound is of relevance when the prior distribution of the mixture
weights $\alpha_h$ is bounded away from zero in a neighborhood of the
boundary of the probability simplex $\Delta_{i-1}$.  Similarly, if the
fiber is restricted to only include points with all $\alpha_h>0$, then
the dimension of this `positive fiber' is only $(i-j)$ and the
analogue of~(\ref{eq:lambdaij-gauss-mix-bound}) is
\begin{equation}
  \label{eq:lca-bound}
  \lambda_{ij} \;\le\; \frac{1}{2}
  \left[\dim(\model_i) - (i-j) \right] = \frac{1}{2}
  \left(ri+j-1 \right) = \bar\lambda_{ij}^r.
\end{equation}
This bound is of interest when the prior distribution of the mixture
weights is zero along the boundary of $\Delta_{i-1}$ but bounded away
from zero in a neighborhood of a point in the positive fiber.
However, as in Sections~\ref{sec:binom-mix} and~\ref{sec:gauss-mix},
we may conclude from the work of \cite{rousseach:2011} that for such
priors it holds that
\begin{equation}
  \label{eq:lca-bound-rouss}
  \lambda_{ij} \;\le\; \frac{1}{2}
  \left[rj+j-1+(i-j)\frac{r}{2} \right] = \bar\lambda_{ij}^{r/2}.
\end{equation}

Contrasting the difference in dimension of the fibers
$\mathcal{F}_i(\pi_0)$ and $\mathcal{F}_i^o(\pi_0)$ when
$\pi_0\in\model_j$ with $j<i$, it is clear that the choice of priors
for the mixture weights $\boldsymbol{\alpha}$ may considerably impact
posterior model probabilities.  In particular, if the prior assigns
non-negligible mass near the boundary of the probability simplex, then
the likelihood function for a sample from $\pi_0\in\model_j$ will be
large near the high-dimensional strata of $\mathcal{F}_i(\pi_0)$.
Model $\model_i$ then behaves like a low-dimensional model, and the
Occam's razor effect from integrating the likelihood function in a
Bayesian approach to model determination is weak.  For our $\SBIC$,
this expresses itself via smaller values of (bounds on) learning
coefficients, which leads to less penalization of the likelihood.
In LCA and similar examples of mixtures of
multi-parameter distributions, it is thus useful to be more explicit
about the effects of priors.

Suppose that the prior distribution
$P(\boldsymbol{\alpha}\,|\,\model_i)$ is a Dirichlet distribution with
all hyperparameters equal to $\phi>0$, and that the remaining
parameters $p_{hl}$ are independent of $\boldsymbol{\alpha}$ a priori
and have a positive joint density on $[0,1]^{ir}$.  Then the learning
coefficients $\lambda_{ij}=\lambda_{i}(\pi_0)$ depend on $\phi$, and
the result in \citet[eqns.~(5) and (6)]{rousseach:2011} shows that the
bounds considered above may be refined to
\begin{align}
  \lambda_{ij} \;\le\; \min\{\bar\lambda_{ij}^\phi,\bar\lambda_{ij}^{r/2}\}.
  \label{eq:lca-phi-bound}
\end{align}
In light of this bound, we let $\overline{\SBIC}_\phi$ denote the
version of our information criterion obtained when using the
$\bar\lambda_{ij}^\phi$ from~(\ref{eq:bar-lambda-phi}) as values of
the learning coefficients and setting all multiplicities to 1.
The behavior of
$\overline{\SBIC}_\phi$ may depend heavily on the choice of $\phi$,
with larger values of $\phi$ leading to stronger penalties and
selection of a smaller number of mixture components.  

When the goal is to stay close to Bayesian inference using
Dirichlet priors for $\boldsymbol{\alpha}$, it may be clear which
value of $\phi$ to use.  It is less clear, however, what a default
choice for $\phi$ should be when $\overline{\SBIC}_\phi$ is intended
to be used as an information criterion with good frequentist model
selection properties.  
Some guidance is provided by Theorem 1 in \cite{rousseach:2011}, which
shows that small enough Dirichlet hyperparameters allow for detection
of zero components in an overfitting mixture model.  According to
their result, working with a single overfitting mixture model can be
an alternative to the model selection setup treated in this paper.
When aiming to determine the number of mixture components in a model
selection approach, however, larger Dirichlet hyperparameters have
appeal in that they avoid large marginal likelihood for models for
whom one or more mixture components will remain empty when using the
model for clustering.  This point is also made in Section 4.2 of the
book by \cite{schnatter:2006}.  More specifically, if we wish to avoid
that overfitting mixture models act like models with fewer components,
then Theorem 1 of \cite{rousseach:2011} suggests that $\phi$ should be
chosen no less than $r/2$.  Given the bound
from~(\ref{eq:lca-phi-bound}), we will thus explore the properties of
$\overline{\SBIC}_\phi$ with $\phi$ close to $r/2$ and compare it to
the standard BIC, which is also equal to $\overline{\SBIC}_{r+1}$.
This said, learning coefficients as large as $\bar\lambda_{ij}^\phi$
with $\phi>r/2$ cannot be realized when the prior
density for the probabilities $p_{hl}$ is everywhere positive but they
could arise from priors whose densities are zero at the singularities
with mixture weights $\alpha_h>0$; see \cite{Petralia:2012} for work
related to this issue.

Our simulations apply BIC and $\overline{\SBIC}_\phi$ for recovery of
the number of classes $i$ in LCA.  We adopt the following four
settings from \cite{nylund:2007} that each have binary items:
\begin{enumerate}[label=(\roman*)]
\item \label{8-simple-equal} $r=8$ items, $i_0=4$ true classes and
  equal class sizes ($\alpha_1=\alpha_2=\alpha_3=\alpha_4$);
\item \label{8-simple-unequal} $r=8$ items, $i_0=4$ true classes and unequal class sizes;
\item \label{10-complex-unequal} $r=10$ items, $i_0=4$ true classes  and unequal class sizes;
\item \label{15-simple-equal} $r=15$ items, $i_0=3$ true classes and equal class sizes
  ($\alpha_1=\alpha_2=\alpha_3$).
\end{enumerate}
Settings \ref{8-simple-unequal} and \ref{10-complex-unequal} have
their unequal class sizes given by $\alpha_1=1/21$, $\alpha_2=2/21$,
$\alpha_3=3/21$, and $\alpha_4=15/21$.  We refer the reader to Table 2
in \cite{nylund:2007} for the precise description of the distributions
we simulate from.\footnote{Our weights $\alpha_i$ are proportional to
  the values $0.05$, $0.1$, $0.15$, and $0.75$, with sum 1.05, that
  are stated in Table 2 of \cite{nylund:2007}.}
Settings~\ref{8-simple-equal} and~\ref{8-simple-unequal} differ only
in the values of the mixture weights/class probabilities $\alpha_h$.
They are both `simple' in the sense that for each item $l$ only one of
the class-conditional probabilities $p_{hl}$ is large.
Setting~\ref{15-simple-equal} is simple in the same sense.  In
setting~\ref{10-complex-unequal}, each item has two large and equal
class-conditional probabilities $p_{hl}$ and the other two
probabilities $p_{hl}$ are small and equal.  For each setting, we draw
100 samples of various sizes $n$ and select the number of classes
$i\in\{1,\dots,6\}$ by maximizing the information criteria.
Specifically, we optimized the standard dimension-based BIC as well as
$\overline{\SBIC}_\phi$ with $\phi=0.5,1,1.5,\dots,r$.  Maximum
likelihood estimates were computed using the R package {\tt poLCA}
\citep{poLCA}.

\begin{table}[t!]
  \centering\scriptsize
  \caption{Latent class analysis:  Frequencies of selection of the
    number of classes by BIC and $\overline{\SBIC}_\phi$ for different values of
    $\phi$.  Four true classes.}
  \label{tab:lca}
  \begin{tabular}{@{\hspace{0.05cm}}l@{\hspace{0.05cm}}r*{3}{r}>{\bfseries}rrrc*{3}{r}>{\bfseries}rrrc*{3}{r}>{\bfseries}rrr}
    \hline\hline
    \multicolumn{2}{c}{} 
    \\[-0.3cm]
    \multicolumn{2}{c}{} & \multicolumn{6}{c}{\emph{Classes}}
    && \multicolumn{6}{c}{\emph{Classes}}&& \multicolumn{6}{c}{\emph{Classes}} \\
    \cline{3-8}  \cline{10-15} \cline{17-22}
    Model & \multicolumn{1}{r}{$n$} &1&2&3&4&5&6&& 1&2&3&4&5&6&& 1&2&3&4&5&6\\
    \hline
    \\[-0.2cm]
    8-item&& \multicolumn{6}{c}{BIC} 
    &\, & \multicolumn{6}{c}{$\overline{\SBIC}_5$}&\, & \multicolumn{6}{c}{$\overline{\SBIC}_{4.5}$} \\
    \cline{3-8}  \cline{10-15} \cline{17-22} 
    \\[-0.3cm]
     
    (equal) 
    & 50 &32&39&26&3  &0&0 &&2&9&21&68 &0&0  &&2&4&14&79&1&0 \\
    & 100&2 &13&32&53 &0&0 &&0&0&3 &97 &0&0  &&0&0&0&100&0&0  \\
    & 150&0 &1 &6 &93 &0&0 &&0&0&0 &99 &1&0  &&0&0&0&98&2&0   \\
    & 200&0 &0 &2 &98 &0&0 &&0&0&0 &99 &1&0  &&0&0&0&97&3&0   \\
    & 500&0 &0 &0 &100&0&0 &&0&0&0 &100&0&0  &&0&0&0&100&0&0  \\
    \cline{3-8}  \cline{10-15} \cline{17-22} 
    \\[-0.1cm]
    \multicolumn{2}{c}{} 
    & \multicolumn{6}{c}{$\overline{\SBIC}_{4}$} 
    &\, & \multicolumn{6}{c}{$\overline{\SBIC}_{3.5}$}&\, & \multicolumn{6}{c}{$\overline{\SBIC}_{3}$} \\
    \cline{3-8}  \cline{10-15} \cline{17-22} 
    \\[-0.3cm]
    &50 & 0&2&9&85&4&0 && 0&0&4&78&18&0 && 0& 0& 1& 61& 29& 9\\
    &100& 0&0&0&96&4&0 && 0&0&0&80&20&0 && 0& 0& 0& 62& 33& 5\\
    &150& 0&0&0&94&5&1 && 0&0&0&87&12&1 && 0& 0& 0& 65& 30& 5\\
    &200& 0&0&0&96&3&1 && 0&0&0&86&10&4 && 0& 0& 0& 78& 17& 5\\
    &500& 0&0&0&97&3&0 && 0&0&0&91& 9&0 && 0& 0& 0& 83& 14& 3\\[0.1cm]
    \hline
   \\[-0.2cm]
    8-item&& \multicolumn{6}{c}{BIC} 
    &\, & \multicolumn{6}{c}{$\overline{\SBIC}_5$}&\, & \multicolumn{6}{c}{$\overline{\SBIC}_{4.5}$} \\
    \cline{3-8}  \cline{10-15} \cline{17-22} 
    \\[-0.3cm]
    
    (unequal) 
    & 100 &9&80&11 &0 &0&0 &&0&21&66&13 &0&0 && 0&12&67&20 &1&0\\ 
    & 200 &0&46&50 &4 &0&0 &&0&3 &61&36 &0&0 && 0&1 &51&47 &1&0\\
    & 300 &0&23&70 &7 &0&0 &&0&0 &31&68 &1&0 && 0&0 &25&73 &2&0\\ 
    & 500 &0&0 &53 &47&0&0 &&0&0 &5 &95 &0&0 && 0&0 &3 &97 &0&0\\
    & 1000&0&0 &2  &98&0&0 &&0&0 &0 &100&0&0 && 0&0 &0 &100&0&0\\
    \cline{3-8}  \cline{10-15} \cline{17-22} 
    \\[-0.1cm]
    \multicolumn{2}{c}{} 
    & \multicolumn{6}{c}{$\overline{\SBIC}_{4}$} 
    &\, & \multicolumn{6}{c}{$\overline{\SBIC}_{3.5}$}&\, & \multicolumn{6}{c}{$\overline{\SBIC}_{3}$} \\
    \cline{3-8}  \cline{10-15} \cline{17-22} 
    & 100 &0&7&58&32 &2&1 &&0&3&45&43&8 &1  &&0&1&23&53&16&7\\
    & 200 &0&1&43&53 &3&0 &&0&0&35&57&8 &0  &&0&0&19&59&22&0\\
    & 300 &0&0&21&76 &3&0 &&0&0&13&76&11&0  &&0&0&5 &72&20&3\\
    & 500 &0&0&1 &98 &1&0 &&0&0&0 &91&9 &0  &&0&0&0 &82&18&0\\
    & 1000&0&0&0 &100&0&0 &&0&0&0 &99&1 &0  &&0&0&0 &91&9 &0 \\[0.1cm]
    
    \hline
    \\[-0.2cm]
    10-item && \multicolumn{6}{c}{BIC} 
    &\, & \multicolumn{6}{c}{$\overline{\SBIC}_{6}$}&\, & \multicolumn{6}{c}{$\overline{\SBIC}_{5.5}$} \\
    \cline{3-8}  \cline{10-15} \cline{17-22} 
    \\[-0.3cm]
    & 100 &0&18&82&0  &0&0 &&0&0&55&44 &1&0  &&0& 0 & 46 & 51 & 3 &0\\
    & 200 &0&1 &84&15 &0&0 &&0&0&23&77 &0&0  &&0& 0 & 16 & 83 & 1 &0\\
    & 300 &0&0 &52&48 &0&0 &&0&0&5 &95 &0&0  &&0& 0 & 3 & 96 & 1 &0 \\
    & 500 &0&0 &11&89 &0&0 &&0&0&1 &99 &0&0  &&0& 0 & 1 & 99 &0& 0  \\
    & 1000&0&0 &0& 100&0&0 &&0&0&0 &100&0&0  &&0& 0 &0& 100 &0& 0   \\
    \cline{3-8}  \cline{10-15} \cline{17-22} 
    \\[-0.1cm]
    \multicolumn{2}{c}{} 
    & \multicolumn{6}{c}{$\overline{\SBIC}_{5}$} 
    &\, & \multicolumn{6}{c}{$\overline{\SBIC}_{4.5}$}&\, & \multicolumn{6}{c}{$\overline{\SBIC}_{4}$} \\
    \cline{3-8}  \cline{10-15} \cline{17-22} 
    & 100 &0&0&31&65 &4&0 &&0&0&22&63&15&0 &&0&0&10&50&32&8 \\
    & 200 &0&0&12&84 &4&0 &&0&0&9 &77&13&1 &&0&0& 6&67&22&5\\
    & 300 &0&0&0 &99 &1&0 &&0&0&0 &94&6 &0 &&0&0& 0&85&14&1\\
    & 500 &0&0&0 &100&0&0 &&0&0&0 &96&4 &0 &&0&0& 0&87&12&1\\
    & 1000&0&0&0 &98 &2&0 &&0&0&0 &98&2 &0 &&0&0& 0&96&3 &1\\[0.1cm]
 
    \hline
    \hline
  \end{tabular}
\end{table}

For settings~\ref{8-simple-equal}-\ref{10-complex-unequal}, our
Table~\ref{tab:lca} reports the frequencies of how often a particular
number of classes $i$ was selected by BIC and $\overline{\SBIC}_\phi$
with $2\phi\in\{r-2,r-1,r,r+1,r+2\}$.  
We see a
tendency for the standard BIC to select overly simple models,
especially at small sample size.  This underselection is alleviated
when using the criteria $\overline{\SBIC}_\phi$ but some overselection
arises for $\phi<r/2$.  The choice $\phi=r/2$ performs quite well, and
so does $\phi=(r+1)/2$.


We do not list any results for small values of $\phi$, such as
$\phi=1$, which corresponds to a uniform distribution as prior for the
mixture weights.  In all but a handful of cases, $\overline{\SBIC}_1$
selected the largest allowed number of classes, that is, $i=6$.  When
the sample size is $n=500$ in setting~\ref{8-simple-equal}, the
relative frequency of $\overline{\SBIC}_\phi$ selecting the truth of
$i_0=4$ classes is 0.01, 0.18, and 0.57 for $\phi=1.5$, $2$, and
$2.5$, respectively.  For $n=1000$ in setting~\ref{8-simple-unequal},
these numbers are 0.02, 0.31, and 0.70.  For $n=1000$ in
setting~\ref{10-complex-unequal}, they are 0.00, 0.00, and 0.14.

\begin{table}[t!]
  \centering\scriptsize
  \caption{Latent class analysis:  Frequencies of selection of the
    number of classes by BIC and $\overline{\SBIC}_\phi$ for different values of
    $\phi$.   Three true classes.}
  \label{tab:lca-15}
  \begin{tabular}{@{\hspace{0.05cm}}lr*{2}{r}>{\bfseries}rrrrc*{2}{r}>{\bfseries}rrrrc*{2}{r}>{\bfseries}rrrr}
    \hline\hline
    \multicolumn{2}{c}{} 
    \\[-0.3cm]
    \multicolumn{2}{c}{} & \multicolumn{6}{c}{\emph{Classes}}
    && \multicolumn{6}{c}{\emph{Classes}}&& \multicolumn{6}{c}{\emph{Classes}} \\
    \cline{3-8}  \cline{10-15} \cline{17-22}
    Model & \multicolumn{1}{r}{$n$} &1&2&3&4&5&6&& 1&2&3&4&5&6&& 1&2&3&4&5&6\\
    \hline
    \\[-0.2cm]
    15-item && \multicolumn{6}{c}{BIC} 
    &\, & \multicolumn{6}{c}{$\overline{\SBIC}_{8.5}$}&\, & \multicolumn{6}{c}{$\overline{\SBIC}_{8}$} \\
    \cline{3-8}  \cline{10-15} \cline{17-22} 
    \\[-0.3cm]
& 50 	&  0&0& 100& 0& 0& 0   &&0&0 &92  &8 &0 &0  &&0&0 &88 &12&0 &0\\ 	  	
& 100 	&  0&0& 100& 0& 0& 0   &&0&0 &98  &2 &0 &0  &&0&0 &93 &7 &0 &0\\ 	  	
& 200 	&  0&0& 100& 0& 0& 0   &&0&0 &99  &1 &0 &0  &&0&0 &97 &3 &0 &0\\ 	   	
& 300 	&  0&0& 100& 0& 0& 0   &&0&0 &100 &0 &0 &0  &&0&0 &99 &1 &0 &0\\ 	
& 400 	&  0&0& 100& 0& 0& 0   &&0&0 &100 &0 &0 &0  &&0&0 &100&0 &0 &0\\ 	
& 500 	&  0&0& 100& 0& 0& 0   &&0&0 &100 &0 &0 &0  &&0&0 &99 &1 &0 &0\\ 		 	
& 1000 	&  0&0& 100& 0& 0& 0   &&0&0 &100 &0 &0 &0  &&0&0 &100&0 &0 &0\\ 		 	
    \cline{3-8}  \cline{10-15} \cline{17-22} 
    \\[-0.1cm]
    \multicolumn{2}{c}{} 
    & \multicolumn{6}{c}{$\overline{\SBIC}_{7.5}$} 
    &\, & \multicolumn{6}{c}{$\overline{\SBIC}_7$}&\, & \multicolumn{6}{c}{$\overline{\SBIC}_{6.5}$} \\
    \cline{3-8}  \cline{10-15} \cline{17-22} 
& 50 	&0&0 &82 &18&0 &0   && 0& 0& 69 &27& 4&0&& 0& 0& 51 &36& 10&3\\ 	
& 100 	&0&0 &81 &19&0 &0   && 0& 0& 68 &32& 0&0&& 0& 0& 53 &43& 4 &0\\	
& 200 	&0&0 &94 &5 &1 &0   && 0& 0& 86 &13& 1&0&& 0& 0& 68 &25& 7 &0\\ 	
& 300 	&0&0 &98 &2 &0 &0   && 0& 0& 96 &4 &0 &0  && 0& 0& 88 &12&0&0\\
& 400 	&0&0 &100&0 &0 &0   && 0& 0& 94 &6 &0 &0  && 0& 0& 87 &13&0&0\\
& 500 	&0&0 &99 &1 &0 &0   && 0& 0& 96 &4 &0 &0  && 0& 0& 91 &9 &0  &0\\
& 1000 	&0&0 &100&0 &0 &0   && 0& 0& 100& 0& 0&0&& 0& 0& 100& 0&0&0\\	 	

    \hline
    \hline
  \end{tabular}
\end{table}

Table~\ref{tab:lca-15} lists the model selection frequencies for the
problem with $r=15$ items and $i_0=3$ true classes.  This is a problem
in which models with $i\le 2$ classes fit so poorly that they are
never selected.  All methods using a heavy penalty thus select the
true number of classes in all cases.  This happens for BIC and
$\overline{\SBIC}_\phi$ with $\phi\ge 11$.  
As earlier, we report details for $2\phi\in\{r-2,r-1,r,r+1,r+2\}$.  We
see overselection for $\phi<r/2=7.5$, which decreases as $\phi$ is
increased to $\phi=r/2=7.5$, $\phi=(r+1)/2=8$, and $\phi=r/2+1=8.5$.
We note that $\overline{\SBIC}_\phi$ with $\phi\le 3$ always selected
the maximum allowed number of classes ($i=6$), and the true number of
classes ($i_0=3$) is never selected when $\phi=3.5$.  When $n=1000$,
the true number of classes ($i_0=3$) is selected with relative
frequency 0.03, 0.33, 0.69, 0.88, and 0.96 when $\phi=4$, $4.5$, $5$,
$5.5$, and $6$, respectively.

In summary, $\overline{\SBIC}_\phi$ can provide considerable
improvements over the standard BIC in terms of frequentist model
selection properties.  To avoid drastic overselection, $\phi$ should
not be chosen too small, compared to $r/2$.  Our above simulations
suggest that taking $\phi=r/2$ or possibly a bit larger, e.g., as
$\phi=(r+1)/2$, could be a good default beyond the specific settings of
latent class analysis that we treated.  

\section{Conclusion}
\label{sec:conclusion}

In this paper we introduced a new Bayesian information criterion for
singular statistical models.  The new criterion, abbreviated $\SBIC$,
is free of Monte Carlo computation and coincides with the widely-used
BIC of Schwarz when the model is regular.  Moreover, the criterion is
consistent and maintains a rigorous connection to Bayesian approaches
even in singular settings.  This latter behavior is made possible by
exploiting theoretical knowledge about the learning coefficients that
capture the large-sample behavior of the concerned marginal likelihood
integrals.  In simulations and data analysis, we showed that $\SBIC$
indeed leads to a `more Bayesian' assessment of model uncertainty and
that it may also lead to improved frequentist model selection when
compared to the standard BIC.  

\subsection*{\rm\em Priors matter for $\SBIC$}

The marginal likelihood of a singular model may depend rather heavily
on the prior distribution.  In fact, the choice of prior may also
have a strong impact on the learning coefficients that quantify the
Occam's razor effect resulting from the integration over parameters.
Therefore, different versions of $\SBIC$, motivated by different
choices of priors, can be of interest for a given singular model
selection problem.

An example where prior distributions play an important role is mixture
modeling with component distributions from a multi-parameter family;
recall our discussion in Section~\ref{sec:numerical-mix}.  Using
latent class analysis for illustration (Section~\ref{subsec:lca}), we
showed how the learning coefficients and thus also $\SBIC$ depend in
particular on whether and how quickly the prior density for the
mixture weights decays to zero or diverges as the weight vector
approaches the boundary of the probability simplex.  We explored this
in the context of Dirichlet prior distributions.  (Strictly speaking,
we considered general bounds for the learning coefficients of mixture
models.)  For good frequentist model selection of $\SBIC$ we suggest
that Dirichlet hyperparameters are not chosen too small.  In
particular, the $\SBIC$ based on a uniform distribution on the mixture
weights cannot be recommended as a default for analyzing mixtures of
multi-parameter distributions.  Similar recommendations for fully
Bayesian approaches to mixture model selection can be found in
\cite{schnatter:2006}.


\subsection*{\rm\em Dependence of $\SBIC$ on the `universe of models'}

When computing the $\SBIC$ of model $\model_i$, we average asymptotic
proxies for the marginal likelihood that are based on Schwarz's idea
of retaining terms from an asymptotic expansion.  The fact that there
is not just a single quantity to contemplate is a feature that
distinguishes singular from regular models.  The terms that are being
averaged correspond to submodels $\model_j\subseteq\model_i$ that are
deemed competitors in the model selection problem.  As a result, the
$\SBIC$ of a singular model $\model_i$ will generally depend on which
set of models we wish to select from.

In most model selection problems there is a canonical set of models to
be considered.  For instance, in mixture modeling one typically
considers all models with up to a certain number of components.  We
envision that $\SBIC$ will generally be applied with respect to such a
natural collection of models, even if the primary focus was on two
specific models.

It is also clear from its definition that the $\SBIC$ of model
$\model_i$ can change only when omitting from consideration a model
$\model_j\subset\model_i$.  We would expect this to be done only if it
is certain that these simpler models are fitting the data poorly,
which would then have little effect on $\SBIC$ scores.  Consider as an
example the version of $\SBIC$ for the galaxies data from
Section~\ref{sec:gauss-mix} (denoted there as $\overline{\SBIC}_1$).
We might wonder how the $\SBIC$ score for $\model_3$ would change if
we no longer considered the too simplistic $\model_1$ and $\model_2$,
which have only 1 and 2 mixture components, respectively.  In the new
context, $\model_3$ would be the minimal model and its $\SBIC$ score
would coincide with the ordinary BIC of $\model_3$.  In Figure 6.3,
the points depicting the BIC and the $\SBIC$ score for $\model_3$
cannot be distinguished.  There is virtually no change in the $\SBIC$
scores when omitting models $\model_1$ and $\model_2$.

Nevertheless, it would be only more appealing if we could define the
$\SBIC$ of a model without reference to the fit of other models.  The
mathematical reason for our consideration of other models is the fact
that our criterion leverages large-sample asymptotics that are based
on fixing a data-generating distribution and letting the sample size
grow.  As is the case in related distribution theory for hypothesis
tests, the limits one obtains will in general not change in a
continuous fashion as we vary the data-generating distribution.  (Of
course, finite-sample behavior of the marginal likelihood will depend
on the data-generating distribution in a continuous fashion.)  Hence,
if we want to avoid consideration of other models in the definition of
a `singular BIC', then more refined mathematical insights would be
necessary.  Specifically, one would need to find uniform asymptotic
expansions to the marginal likelihood, in the sense of \citet[Chapter
VII]{Wong}.  This, however, is a task that would be significantly
harder to accomplish than finding the already non-trivial to obtain
learning coefficients.  Indeed, we are not aware of any discussion of
uniform expansions in the statistical literature, let alone any
results on their form for specific examples.  In light of these
difficulties, we consider our proposed $\SBIC$ a promising approach of
averaging point-wise expansions to mimic how uniform expansions would
have to behave.

\subsection*{\rm\em Large numbers of models}

For problems that involve a moderate number of models and are amenable
to an exhaustive model search, the computational effort in the
calculation of $\SBIC$ scores is comparable to that for the ordinary
BIC as the effort is typically dominated by the process of fitting all
considered models to the available data.  However, the fact that our
definition of $\SBIC$ requires fitting all considered models has a
clear computational disadvantage when an exhaustive search is not
possible.  Indeed, it is not immediately clear how to implement
strategies such as greedy search with $\SBIC$.  One possible approach
would be to define $\SBIC$ by averaging only over `neighboring'
submodels but the merit of such strategies still needs to be explored.
This said, the work of \cite{drton:lin:weihs:zwiernik:2014} shows
promising results for selection of Gaussian latent forest models.

We note that when treating problems with a large number of models it
can be beneficial to adopt non-uniform prior model probabilities;
compare e.g.\ the work on regression models by \cite{Chen:2008} and
\cite{Scott:2010}, and the work on graphical models by
\cite{Foygel:2010} and \cite{Gao:2012}.  As mentioned in
Remark~\ref{rem:nonuniform}, it would be straightforward to
incorporate prior model probabilities into the definition of $\SBIC$.

\subsection*{\rm\em Use of maximum likelihood estimates}

Our aim was to generalize Schwarz's BIC in a way that recovers his
familiar criterion when the considered models are regular (recall
Remark~\ref{prop:sbic=bic}).  To this end, we estimate the true
likelihood by evaluating the likelihood function at the maximum
likelihood estimator.  However, other estimators could be used
instead.  For instance, \cite{Roeder:1997} used posterior means.
Similarly, one could consider posterior modes/penalized likelihood
methods to stay closer to a fully Bayesian analysis or simply for
regularization; see \cite{Fraley:2007} and \cite{Baudry:2015} for work
on Gaussian mixtures.  We note that penalization of the likelihood
function would provide a way to address the failure of
assumption~\ref{A1} from Section~\ref{sec:large-sample} that may occur
in mixture models with unbounded parameter space
\citep{hartigan:1985}.

\subsection*{\rm\em When learning coefficients are not known}

To our knowledge, $\SBIC$ is the first statistical method to make use
of mathematical information about the values of learning coefficients
of singular models.  The theoretical insights allow one to obtain
(crude) approximations to posterior model probabilities without Monte
Carlo integration.  At the same time, the reliance on theory also
presents a limitation as the learning coefficients may not always be
known.  Previous studies have shown that when exact values of learning
coefficients are difficult to find, it may still be possible to obtain
bounds.  For priors that are bounded from above, a learning
coefficient can be trivially bounded by the model dimension and using
dimensions in $\SBIC$ recovers the standard BIC (recall
Remark~\ref{prop:sbic=bic}).  However, more interesting bounds
can often be found by arguments that are only slightly more
complicated than parameter counting.  The usefulness of such bounds
was demonstrated in Section~\ref{sec:numerical-mix}.

Finally, our $\SBIC$ provides strong positive motivation for
theoretical studies of learning coefficients.  From a statistical
perspective, past work had a negative flavor; knowing the
values one could stress just how much smaller they can be than a
parameter count.  In contrast, new theoretical insights now yield new
statistical methodology.  We anticipate that this positive motivation
will lead to further work and results on learning coefficients.

\section*{Acknowledgments}

This collaboration started at a workshop at the American Institute of
Mathematics, and we would like to thank the participants of the
workshop for helpful discussions.  Particular thanks go to Vishesh
Karwa, Dennis Leung and Luca Weihs for help with some of the numerical
work.  Mathias Drton was supported by grants from the NSF
(DMS-1305154) and the RRF at the University of Washington
as well as an Alfred P. Sloan Fellowship.

\bibliographystyle{rss}
\bibliography{singular_bic}

\end{document}